\documentclass[sn-mathphys]{sn-jnl}

\usepackage[T1]{fontenc}
\usepackage[utf8]{inputenc}
\usepackage{amsmath,amssymb,amsfonts,amsthm}
\usepackage{mathtools}
\usepackage{algorithm}
\usepackage{algpseudocode}
\usepackage{etoolbox}
\usepackage{graphicx}
\usepackage{textcomp}
\usepackage{xcolor}
\usepackage{placeins}
\usepackage{enumitem}
\usepackage{multirow}
\usepackage{makecell}

\makeatletter
\def\algbackskip{\hskip-\ALG@thistlm}
\newcommand*{\algrule}[1][\algorithmicindent]{%
    \makebox[#1][l]{%
        \hspace*{.2em}
        \vrule height .75\baselineskip depth .25\baselineskip
    }
}

\newcount\ALG@printindent@tempcnta
\def\ALG@printindent{%
    \ifnum \theALG@nested>0
    \ifx\ALG@text\ALG@x@notext
    \else
    \unskip
    \ALG@printindent@tempcnta=1
    \loop
    \algrule[\csname ALG@ind@\the\ALG@printindent@tempcnta\endcsname]%
    \advance \ALG@printindent@tempcnta 1
    \ifnum \ALG@printindent@tempcnta<\numexpr\theALG@nested+1\relax
    \repeat
    \fi
    \fi
}
\patchcmd{\ALG@doentity}{\noindent\hskip\ALG@tlm}{\ALG@printindent}{}{\errmessage{failed to patch}}
\patchcmd{\ALG@doentity}{\item[]\nointerlineskip}{}{}{} 
\makeatother

\usepackage{subcaption}

\usepackage[noabbrev]{cleveref}

\makeatletter

\theoremstyle{thmstyleone}%
\newtheorem{thm}{Theorem}
\newtheorem{prop}[thm]{Proposition}

\theoremstyle{thmstylethree}%
\newtheorem{defn}[thm]{Definition}
\theoremstyle{thmstyletwo}%
\newtheorem{rem}[thm]{Remark}

\makeatother

\jyear{2022}%

\providecommand{\lemmaname}{Lemma}

\Crefname{cor}{Corollary}{Corollaries}
\Crefname{prop}{Proposition}{Propositions}

\newcommand{\conv}[3]{#1 \otimes #2 \, (#3)}

\newcommand{\posPart}[1]{\left[#1\right]^+}

\newcommand{\card}[1]{n\!\left(#1\right)}

\newcommand{\lpi}[1]{#1^{-1}_{\downarrow}}
\newcommand{\lpib}[1]{\left(#1\right)^{-1}_{\downarrow}}

\newcommand{\upi}[1]{#1^{-1}_{\uparrow}}
\newcommand{\upib}[1]{\left(#1\right)^{-1}_{\uparrow}}

\newcommand{\genop}[1]{#1^*}

\begin{document}

\title[Extending the NC Algorithmic Toolbox for UPP Functions]{Extending the Network~Calculus Algorithmic Toolbox for Ultimately~Pseudo-Periodic~Functions: Pseudo-Inverse and Composition}


\author*[1,2,3]{\fnm{Raffaele} \sur{Zippo}}\email{raffaele.zippo@unifi.it}

\author[3]{\fnm{Paul} \sur{Nikolaus}}\email{nikolaus@cs.uni-kl.de}

\author[2]{\fnm{Giovanni} \sur{Stea}}\email{giovanni.stea@unipi.it}

\affil*[1]{\orgdiv{Dipartimento di Ingegneria dell'Informazione}, \orgname{Universit{\`a} di Firenze}, \orgaddress{\street{Via di S. Marta 3}, \city{Firenze}, \postcode{50139}, \country{Italy}}}

\affil[2]{\orgdiv{Dipartimento di Ingegneria dell'Informazione}, \orgname{Universit{\`a} di Pisa}, \orgaddress{\street{Largo Lucio Lazzarino 1}, \city{Pisa}, \postcode{56122}, \country{Italy}}}

\affil[3]{\orgdiv{Distributed Computer Systems Lab (DISCO)}, \orgname{TU Kaiserslautern}, \orgaddress{\street{Paul-Ehrlich-Straße 34}, \city{Kaiserslautern}, \postcode{67663}, \country{Germany}}}

    
\abstract{
	Network Calculus (NC) is an algebraic theory that represents traffic and service guarantees as curves in a Cartesian plane, in order to compute performance guarantees for flows traversing a network.
	NC uses transformation operations, e.g., min-plus convolution of two curves, to model how the traffic profile changes with the traversal of network nodes.
	Such operations, while mathematically well-defined, can quickly become unmanageable to compute using simple pen and paper for any non-trivial case, hence the need for algorithmic descriptions.
	Previous work identified the class of piecewise affine functions which are ultimately pseudo-periodic (UPP) as being closed under the main NC operations and able to be described finitely. Algorithms that embody NC operations taking as operands UPP curves have been defined and proved correct, thus enabling software implementations of these operations.
	
	However, recent advancements in NC make use of operations, namely the \emph{lower pseudo-inverse}, \emph{upper pseudo-inverse}, and \emph{composition}, that are well defined from an algebraic standpoint, but whose algorithmic aspects have not been addressed yet.
	In this paper, we introduce algorithms for the above operations when operands are UPP curves, thus extending the available algorithmic toolbox for NC. We discuss the algorithmic properties of these operations, providing formal proofs of correctness.
}


\maketitle

\section{Introduction}
\label{sec:Introduction}
Network Calculus (NC) is an algebraic theory where traffic and service guarantees are represented as functions of time. 
The I/O transformations that network traversal imposes on an input traffic can be represented as operations of min-plus algebra involving these curves. 
This allows one to compute worst-case performance guarantees for a flow traversing a network.
NC dates back to the early 1990s, and it is mainly due to the work of Cruz \cite{Cruz1991NetCal-I, Cruz1991NetCal-II}, Le Boudec and Thiran \cite{LeBoudec2001NetCal}, and Chang \cite{Chang2000Performance}.
Originally devised for the Internet, where it was used to engineer models of service \cite{le1998application,firoiu2002theories, 669170, bennett2002delay, fidler2004parameter}, it has found applications in several other contexts, from sensor networks \cite{schmitt2005sensor} to avionic networks \cite{charara2006methods, bauer2010worst}, industrial networks \cite{zhang2019analysis,maile2020network,BoyerTSN2021} automotive systems \cite{rehm2021road} and systems architecture \cite{andreozzi2020heterogeneous,BoyerGDM20}.

NC characterizes constraints on traffic arrivals (due to traffic shaping) and on minimum received service (due to scheduling) as \emph{curves}, i.e., functions of time.\footnote{We use the terms \emph{function} and \emph{curve} interchangeably}
These curves are then used with operators from min-plus and max-plus algebra to derive further insights about the system.
For example, the per-flow service curve of a scheduler, such as Weighted Round Robin, or performance bounds on the traffic such as an end-to-end delay bound.

While these operations can be computed with pen and paper for simple examples, in most practical cases the application of NC requires the use of software. 
To this end, \cite{bouillard2008algorithmic,Bouillard2018DNC} provide an ``algorithmic toolbox'' for NC: 
they show that piecewise affine functions that are ultimately pseudo-periodic (UPP) represent good models for both traffic and service guarantees. 
Moreover, they prove that this class of functions is closed under the main NC operations and can be described with a finite amount of information. 
Additionally, they introduce the algorithms that embody the main NC operations, computing UPP results starting from UPP operands. 
The results in these works cover the main operations used in NC, such as minimum, min-plus convolution, min-plus deconvolution, etc. (see \cite{Bouillard2018DNC} for a complete list).
The toolbox was first implemented in the COINC free library \cite{bouillard2009coinc}, which is not available anymore, and later by the commercial library RTaW-Pegase \cite{RTaW-Pegase-Library}.

However, other NC operators, i.e., the composition and lower and upper pseudo-inverse, have been in focus in recent NC literature. 
\cite{Liebeherr2017} shows the duality between min-plus and max-plus models, and how the lower and upper pseudo-inverses can be used to switch between the two models.
In \cite[Theorem~8.6]{Bouillard2018DNC}, lower pseudo-inverse and composition are used to compute the per-flow service curve for a Weighted Round-Robin scheduler; 
in \cite[Theorem~1]{Tabatabaee2021IWRR}, a similar result is shown for a Interleaved Weighted Round-Robin scheduler, using again lower pseudo-inverse and composition; 
in \cite{Tabatabaee2022DRR}, authors show that several service curves can be found for a flow scheduled in a Deficit Round-Robin scheduler, under different hypotheses regarding cross-traffic. 
\cite{mohammadpour2019DelayKnownRate, mohammadpour2022improvedNCinTSN} use pseudo-inverses and compositions to study properties of IEEE Time-Sensitive Networking (TSN) \cite{TSNTaskGroup}, a standard relevant for many applications.
We can therefore obtain the results presented in these papers, using arbitrarily complex UPP curves as inputs.

While the algebraic formulation of these three operations is well known, their algorithmic aspects have not been addressed, to the best of our knowledge. 
This means that we do not have publicly known algorithms that compute these operations yet.
 
In this paper, we aim to fill this gap and extend the existing algorithm toolbox to include lower- and upper-pseudo inverses and composition of functions.

We show that the UPP class is closed with respect to these operations, and provide algorithms to compute the result of each one. 
We prove that all of them have linear complexity with respect to the number of segments that represent the operands. 
We design specialized, more efficient versions of the composition algorithm that leverage characteristics of the operand functions -- notably, their being Ultimately Affine (UA) or Ultimately Constant (UC) \cite{Bouillard2018DNC}. 
Last, we exemplify our findings on a comprehensive proof of concept, showing how to compute the per-flow service curve of \cite[Theorem~1]{Tabatabaee2021IWRR}, which makes use of all the findings in this paper. 
The algorithms described in this paper, together with those for known NC operators, are implemented in the Nancy open-source toolbox \cite{Nancy22}, which, to the best of our knowledge, is the only public one to implement UPP algorithms.

The rest of this paper is organized as follows:
\Cref{sec:NcBasics} briefly introduces Network Calculus notation and some basic results.
\Cref{sec:MathBackground} introduces the definitions and notation used throughout the paper, and discusses the kind of results that we need to provide for each operator to enable their implementation.
In \Cref{sec:PseudoInverses}, we present the results for the lower and upper pseudo-inverse operators, including their properties for UPP curves and algorithms to compute them.
\Cref{sec:Composition} shows our results for the composition operator, including its properties on UPP curves and an algorithm to compute it.
In \Cref{sec:PoC}, we report a proof-of-concept evaluation, by computing the results of a recent NC paper via our algorithms.
Finally, \Cref{sec:Conclusions} draws some conclusions and highlights directions for future works.

\section{Network Calculus Basics}
\label{sec:NcBasics}

This section briefly introduces Network Calculus (NC). 
A NC flow is represented by a function of time $A(t)$ that counts the amount of traffic arrived by time $t$. Such function is necessarily wide-sense increasing. 
It is often assumed to be left-continuous, i.e., $A(t)$ represents the number of bits in $[0,t[$. 
In particular, $A(0)=0$.

Flows can be constrained by \emph{arrival curves}. A wide-sense increasing function $\alpha$ is an \emph{arrival curve} for a flow  $A$ if
\begin{equation*}
    A(t)-A(s) \leq \alpha(t-s), \qquad \forall s \leq t.
\end{equation*}
For instance, a \emph{leaky-bucket shaper}, with a \emph{rate} $\rho$ and a \emph{burst size} $\sigma$, enforces the concave affine arrival curve $\gamma_{\sigma,\rho}(t) = \sigma + \rho t$, as shown in \Cref{fig:leaky-bucket}. 
In particular, this means that the long-term arrival rate of the flow cannot exceed $\rho$. 
Leaky-bucket shapers are often employed at the entrance of a network, to ensure that the traffic injected does not overflow the negotiated amount. 

\begin{figure}[!t]
    \centering
    \includegraphics[width=.55\linewidth]{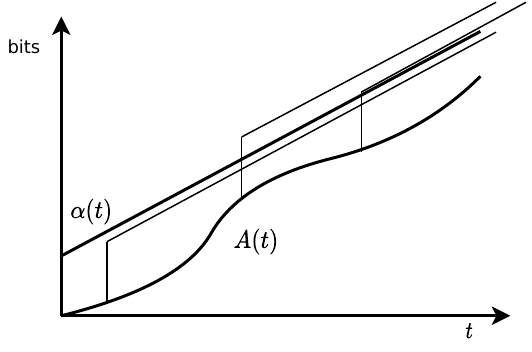}
    \caption{
    Example of leaky-bucket shaper, taken from \cite{andreozzi2020heterogeneous}. 
    The traffic process $A(t)$ is always below the arrival curve $\alpha(t)$ and its translations along $A(t)$.}
    \label{fig:leaky-bucket}
\end{figure}

\begin{figure}[!t]
    \centering
    \includegraphics[width=.45\linewidth]{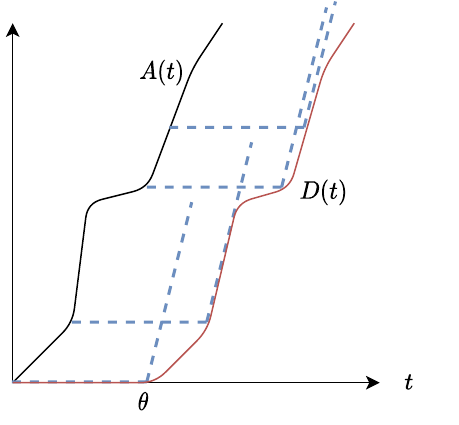}
    \caption{Graphical interpretation of the convolution operation.}
    \label{fig:convolution}
\end{figure}

\begin{figure}[!t]
    \centering
     \includegraphics[width=.55\linewidth]{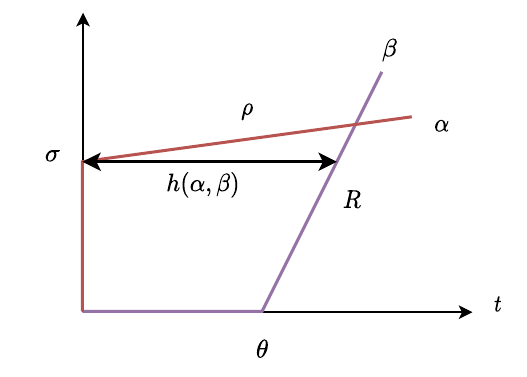}
    \caption{Graphical example of a delay bound.}
    \label{fig:delay-bound}
\end{figure}

Let $A$ and $D$ be non-decreasing functions that describe the same data flow at the input and output of a lossless network element (or \emph{node}), respectively.
If that node does not create data internally (which is often the case), causality requires that $A\geq D$.
We say that the flow is guaranteed the (minimum) \emph{service curve} $\beta$ if 
\begin{equation}
    \label{eq:conv}
        D(t) \geq \inf_{0\leq s\leq t} A(s)+\beta(t-s)
        \eqqcolon (A \otimes \beta)(t), \qquad \forall t\geq 0.
\end{equation}
We call the operation on the right the \emph{(min-plus) convolution} of $A$ and $\beta$.
The function argument $t$ is omitted whenever it is clear from the context.
Several network elements, such as delay elements, schedulers or links, can be modeled through service curves.

A very frequent case is the one of \emph{rate-latency} service curves, defined as
\begin{equation*}
    \beta_{R, \theta}(t) = R \cdot \posPart{t-\theta}
\end{equation*}
for some rate $R>0$ and latency $\theta>0$.
Notation $\posPart{\cdot}$ denotes $\max\left\{\cdot, 0\right\}$.
For instance, a constant-rate server (e.g., a wired link) can be modeled as a rate-latency curve with zero latency. 
\Cref{fig:convolution} shows the lower bound of $D(t)$ obtained when convolving the input $A(t)$ with a rate-latency service curve.

A point of strength of NC is that service curves are \emph{composable}: the end-to-end service curve of a tandem of nodes traversed by the same flow can be computed as the convolution of the service curves of each node.

For a flow that traverses a service curve (be it the one of a single node, or the end-to-end service curve of a tandem computed as discussed above), a tight \emph{upper bound} on the delay can be computed by combining its arrival curve $\alpha$ and the service curve $\beta$ itself, as follows:
\begin{equation}
    \label{eq:delaybound}
    h(\alpha,\beta) = \sup_{t\geq 0} \left[\inf \left\{ d\geq 0 \mid \alpha(t-d)\leq \beta(t) \right\} \right].
\end{equation}

The quantity $h(\alpha,\beta)$ is in fact the maximum horizontal distance between $\alpha$ and $\beta$, as shown in \Cref{fig:delay-bound}. 
Therefore, computing the end-to-end service curve of a flow in a tandem traversal is the crucial step towards obtaining its worst-case delay bound. 

The above introduction, albeit concise, should convince the alert reader that algorithms for automated manipulation of curves, implementing NC operators, are necessary to reap the benefits of NC algebra in practical scenarios. 
Many such algorithms have been discussed in \cite{bouillard2008algorithmic, Bouillard2018DNC}. 
The next section describes the generic algorithmic framework exposed in these papers, which we extend in this work.

\section{Mathematical Background and Notation}
\label{sec:MathBackground}

In this section, we provide an overview of the mathematical background for this paper, including the definitions used and the results we aim to provide.

NC computations can be implemented in software.
In order to do so, one needs to provide finite representations of functions and well-formed algorithms for NC operations.
According to the widely accepted approach described in \cite{bouillard2008algorithmic, Bouillard2018DNC}, a sufficiently generic class of functions useful for NC computations is that of (i) piecewise affine (ii) ultimately pseudo-periodic $\mathbb{Q_+} \rightarrow \mathbb{Q}$ functions. 
We define both properties (i) and (ii) separately:

\begin{defn}[Piecewise Affine Function {\cite[p.~9]{bouillard2008algorithmic}}]
    We say that a function $ f $ is \emph{piecewise affine} (PA) if there exists an increasing sequence $(a_{i})_{i\in\mathbb{N}}$ which tends to $+\infty$, such that $a_{0}=0$ and $\forall i\geq0,$ $f$ is affine on $\left]a_{i},a_{i+1}\right[$i.e.,
    \begin{equation*}
        f(t)\in\left\{ b_{i}+\rho_{i}t,+\infty,-\infty\right\} ,\qquad\forall t\in\left]a_{i},a_{i+1}\right[.
    \end{equation*}
    The $a_{i}$'s are called \emph{breakpoints}.
\end{defn}

\begin{defn}[Ultimately Pseudo-Periodic Function {\cite[p.~8]{bouillard2008algorithmic} }]
    Let $f$ be a function $\mathbb{Q}^{+}\rightarrow\mathbb{Q}  \cup \{+\infty, -\infty\}$.
    Then, $f$ is ultimately pseudo-periodic (UPP) if there exist $T_{f}\in\mathbb{Q}^{+},d_{f}\in\mathbb{Q}^{+}\setminus\{0\},c_{f}\in\mathbb{Q}$ such that 
    \begin{equation}
        f(t+k_{f}\cdot d_{f})=f(t)+k_{f}\cdot c_{f}, \qquad\forall t\ge T_{f},\forall k_{f}\in\mathbb{N}_{0}.
        \label{eq:UPP-property}
    \end{equation}
    We call $T_{f}$ the (pseudo-periodic) start or length of the initial transient, $d_{f}$ the (pseudo-periodic) length, and $c_{f}$ the (pseudo-periodic) height.
\end{defn}

In \cite{bouillard2008algorithmic}, this class of functions is shown to be stable w.r.t. all min-plus operations, while functions $\mathbb{R_+} \rightarrow \mathbb{R}$ are not.\footnote{ An alternative class of functions with such stability is $\mathbb{N} \rightarrow \mathbb{R}$, however this is only feasible for models where time is discrete.}
Henceforth, we will consider all functions to be piecewise affine and UPP.
For such functions, it is enough to store a representation of the initial transient part and of one period, which is a finite amount of information.
This is exemplified in Figure \ref{fig:PseudoPeriodicExample}.

\begin{figure}[!t]
    \centering
    \begin{subfigure}{0.45\linewidth}
        \includegraphics[width=\linewidth]{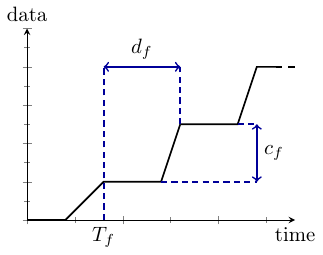}
        \caption{$f$}
        \label{fig:PseudoPeriodicExample.f}
    \end{subfigure}
    \hspace*{\fill}
    \begin{subfigure}{0.45\linewidth}
        \includegraphics[width=\linewidth]{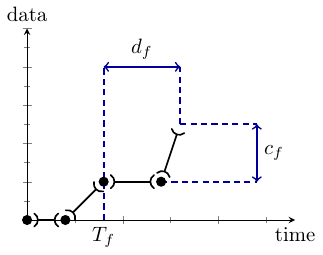}
        \caption{$R_f$}
        \label{fig:PseudoPeriodicExample.Rf}
    \end{subfigure}

    \caption{Example of ultimately pseudo-periodic piecewise affine function $f$ and its representation $R_f$. 
    }
    \label{fig:PseudoPeriodicExample}
\end{figure}

Accordingly, we call a \emph{representation} $R_f$ of a function $f$ the tuple $(S, T, d, c)$, where $T, d, c$ are the values described above, and $S$ is a sequence of points and open segments describing $f$ in $[0, T + d[$.
We use both points and open segments in order to easily model discontinuities. 
We will use the umbrella term \textit{elements} to encompass both when convenient.

\begin{defn}[Point]
    We define a \emph{point} as a tuple
    \begin{equation*}
        p_i \coloneqq (t_{i}, f(t_{i})), \qquad i \in \left\{1, \dots, n\right\}.
    \end{equation*}
\end{defn}

\begin{defn}[Segment]
    We define a \emph{segment} as a tuple
    \begin{equation*}
        s_i \coloneqq \left(t_{i}, t_{i+1}, f(t_{i}^{+}), f(t_{i+1}^{-})\right), \qquad i \in \left\{1, \dots, n\right\}
    \end{equation*} 
    which describes $ f $ in the open interval $ \left]t_{i}, t_{i+1} \right[ $ in which it is affine, i.e., 
    \begin{equation*} 
        f(t) = f(t_{i}^+) + \frac{f(t_{i+1}^-) - f(t_{i}^+)}{t_{i+1} - t_{i}} \cdot (t - t_{i}) \eqqcolon b + r \cdot (t - t_{i}) \qquad \text{for all } t \in \left]t_{i}, t_{i+1} \right[.
    \end{equation*}
    If $ r = 0, $ we call $ s_i $ a \emph{constant segment}.
\end{defn}

\begin{defn}[Sequence]
    We define a sequence $S^D_f$ as on ordered set of elements ${e_1, \dots, e_n}$ that alternate between points and segments and describe $ f $ in the finite interval $D$.
\end{defn}

For example, given $D = \left[ 0, T \right[$, then $S^D_f = \{p_1, s_1, p_2, \dots, p_n, s_n\}$ where $p_1 = \left(0, f(0)\right)$ and, assuming $p_n = \left(t_n, f(t_n)\right)$ for some $0 < t_n < T$, $s_n = \left(t_n, T, f(t_n^+), f(T^-)\right)$.

Note that, given $R_f$, one can compute $f(t)$ for all $t \ge 0$, and also $S^D_f$ for any interval $D$. 
Furthermore, being finite, $R_f$ can be used as a data structure to represent $f$ in code. 
As is discussed in depth in \cite{ZippoEtAlii2022}, $R_f$ is not unique, and using a \emph{non-minimal} representation of $f$ can affect the efficiency of the computations.
Given a sequence $S$, let $\card{S}$ be its cardinality.
As it is useful in the following, we define \emph{Cut} to be an (obvious) algorithm that, given $R_f$ and $D$, computes $S^D_f$.
With a little abuse of notation, we use min-plus operators directly on finite sequences such as $S^D_f$. 
For instance, given the lower pseudo-inverse of $f$, $\lpib{f}$ (its formal definition is in the next section), we will write $\lpib{S^D_f}$, to express that we are computing it on $f$ over the limited interval $D$.

A NC operator can then be defined computationally as an algorithm that takes UPP representations of its input functions and yields an UPP representation of the result, provided that the class of UPP functions is closed with respect to such operator.
Considering a generic unitary operator $\genop{[\cdot]}$, in order to compute $\genop{f}$ we need an algorithm that computes $R_{\genop{f}}$ from $R_f$, i.e., $R_f \rightarrow R_{\genop{f}}$.
We call this \emph{by-curve} algorithm.
This process can be divided in the following steps:
\begin{enumerate}
    \item compute valid parameters $T_{\genop{f}}, d_{\genop{f}}$ and $c_{\genop{f}}$ for the result;
    \item compute the domain $D$ for the sequence $S^{D}_f = \text{Cut}(R_f, D)$;
    \item compute $S^{D}_f \rightarrow S_{\genop{f}}$, i.e., use an algorithm that computes the resulting sequence from the sequence of the operand. 
        We call this \emph{by-sequence} algorithm for operator $\genop{[\cdot]}$;
    \item return $R_{\genop{f}} = (S_{\genop{f}}, T_{\genop{f}}, d_{\genop{f}}, c_{\genop{f}})$.
\end{enumerate}

We therefore need to provide the following results:
\begin{itemize}
    \item a proof that the result of the operator $\genop{[\cdot]}$, applied to an UPP function, yields an UPP result;
    \item a way to compute UPP parameters $T_{\genop{f}}, d_{\genop{f}}$ and $c_{\genop{f}}$ from $R_f$;
    \item a valid domain $D$, again to be computed from $R_f$;
    \item a \emph{by-sequence} algorithm.
\end{itemize}

Combining the above results, we can then construct the \emph{by-curve} algorithm for operator $\genop{[\cdot]}$, 
thus allowing us to compute $\genop{[\cdot]}$ for any UPP curve.\footnote{ The same process applies also, with minor adjustments, to binary operators}

Works \cite{bouillard2008algorithmic, Bouillard2018DNC} provided such computational descriptions for fundamental NC operators such as minimum, sum, convolution, and many others.
In this paper, we extend the above toolbox by adding the lower pseudo-inverse, upper pseudo-inverse, and composition operators.
To the best of our knowledge, no computational description of the above has been formalized before, despite their relevance in the NC literature.

Before presenting our contribution, we introduce two more definitions that will be used throughout the paper.

\begin{defn}[Ultimately Affine Function]
    \label{def:ultimately-affine}
    Let $f$ be a function $\mathbb{Q}^{+}\rightarrow\mathbb{Q}$. 
    Then, $ f $ is Ultimately Affine (UA), if there exist $T_{f}^{a}\in\mathbb{Q}^{+},\rho_{f}\in\mathbb{Q}$
    such that 
    \begin{equation}
        f(t) = f(T_{f}^{a}) + \rho_{f} \cdot \left(t - T_{f}^{a}\right), \qquad\forall t \ge T_{f}^{a}.
        \label{eq:ultimately-affine}
    \end{equation}
\end{defn}

Note that this definition differs from the one in the literature \cite{bouillard2008algorithmic}, but we prove their equivalence in \Cref*{app-sec:properties-UAF}. 
UA functions are (obviously) UPP as well, their period being a single segment of slope $\rho_f$ and arbitrary length starting at $T_{f}^{a}$. 
They occur quite often in NC, e.g., the arrival curve of a leaky-bucket shaper or a rate-latency service curve are both UA. 
An \emph{Ultimately Constant} (UC) function is UA with $\rho_f = 0$.

Unlike UPP, the class of UA functions is not closed with respect to NC operations. 
For instance, \cite{ZippoEtAlii2022} shows that flow-controlled networks with rate-latency (hence UA) service curves yield closed-loop service curves that are UPP, but not necessarily UA again. 
Moreover, in many cases, the service curves of individual flows served by Round-Robin schedulers are UPP, but not UA either (see, e.g., \cite{Boyer2012, Tabatabaee2021IWRR, Tabatabaee2022DRR}). 
However, there are cases when simpler algorithms for NC operations can be derived if one assumes that operands are UA. 
For this reason, there are NC toolboxes that only consider UA functions, e.g., \cite{DISCO}. 
A possible approach to NC analysis is thus to approximate UPP (non-UA) functions with UA lower/upper bounds, trading some accuracy for computation time \cite{guan2013finitaryRTC, lampka17GeneralFinitaryRTC}.
Throughout this paper, we provide general algorithms for UPP functions. 
However, we also show what is to be gained -- in terms of domain compactness and/or algorithmic efficiency -- when we can make stronger assumptions on the operands.

\begin{defn}[Weakly Ultimately Infinite Function]
    \label{def:wUI-function}
    Let $ f:\mathbb{Q}^+ \to \mathbb{Q}^+ \cup \{+\infty\} $ be an UPP function. 
    If there exists $ T_f \in \mathbb{Q}^+ $ such that
    \begin{align*}
        f(t) &< +\infty, \qquad \forall t < T_f,\\
        f(t) &= +\infty, \qquad \forall t > T_f,
    \end{align*}
    then we call this function a \emph{weakly ultimately infinite} (wUI) one. 
\end{defn}

Weakly ultimately infinite functions are infinite starting from a finite abscissa. 
Typical cases in NC are the service curves of \emph{delay elements}.
Note that the above definition differs from the one of \emph{ultimately infinite} in \cite{bouillard2008algorithmic}: the latter requires $f$ to be infinite \emph{in} $T_f$ (i.e., it has a weak inequality in the lower branch), whereas we do not mandate anything in $T_f$. 
Therefore, every ultimately infinite function is wUI too, but not vice versa.
Our more general definition allows us to better discuss corner cases in \Cref{subsec:PseudoInverses-CornerCases}.

\section{Lower and Upper Pseudo-Inverse of UPP Functions}
\label{sec:PseudoInverses}

In this section, we discuss the lower and upper pseudo-inverse operators for UPP functions (henceforth, we will omit the \emph{lower} or \emph{upper} attribute when the discussion applies to both). 

First, we provide formal definitions.

\begin{defn}[Lower and Upper Pseudo-Inverse \cite{Liebeherr2017}]
    \label{def:PseudoInverses}
    Let $f:\mathbb{Q}\to\mathbb{Q}\cup\{-\infty\}\cup\{\infty\}$ be
    non-decreasing. Then the \emph{lower pseudo-inverse} is defined as
    \begin{equation*}
    \lpi{f}(y)\coloneqq\inf\left\{ x\mid f(x)\geq y\right\} =\sup\left\{ x\mid f(x)<y\right\} 
    \end{equation*}
    and the \emph{upper pseudo-inverse} is defined as
    \begin{equation*}
    \upi{f}(y)\coloneqq\sup\left\{ x\mid f(x)\leq y\right\} =\inf\left\{ x\mid f(x)>y\right\} .
    \end{equation*}
\end{defn}

Note that the lower pseudo-inverse is always left-continuous and the upper pseudo-inverse is right-continuous \cite[pp. 64]{Liebeherr2017}.
Moreover, we have in general that \cite[pp. 61]{Liebeherr2017}
\begin{equation*}
    \lpi{f}\leq \upi{f}.
\end{equation*}

The rest of this section is organized as follows.
In \Cref{subsec:PseudoInverses-Theorems} we show that the pseudo-inverse of an UPP function is still UPP, and provide expressions to compute its UPP parameters a priori.
In \Cref{subsec:PseudoInverses-Algorithm} we discuss, first through an visual example and then via pseudo-code, how to algorithmically compute the pseudo-inverse.
In \Cref{subsec:PseudoInverses-Conclusions} we conclude with a summary of the by-curve algorithm, some observations on the algorithmic complexity of this operator.
In \Cref{subsec:PseudoInverses-CornerCases} we discuss corner cases.

\subsection{Properties of pseudo-inverses of UPP functions}
\label{subsec:PseudoInverses-Theorems}

We discuss our properties for a generic function $f$, \emph{excluding} the cases of UC and wUI functions. 
These two cases are treated separately for ease of presentation.
At the end of this section, we report the necessary information for the alert reader to retrace the steps exposed hereafter to include these two corner cases. 
We remark that the Nancy software library \cite{Nancy22} computes pseudo-inverses of generic (non-decreasing) UPP functions, including UC and wUI ones.

\begin{thm}
\label{th:UppLowerPseudoInverse}
    Let $f$ be a non-decreasing UPP function that is neither UC nor wUI.
    Then, its lower pseudo-inverse $\lpi{f}(x)=\inf\left\{ t\mid f(t)\ge x\right\} $ is again UPP with  
    \begin{align}
        T_{\lpi{f}} & =f\left(T_{f}+d_{f}\right),\\
        d_{\lpi{f}} & =c_{f}\label{eq:lpi-d},\\
        c_{\lpi{f}} & =d_{f}\label{eq:lpi-c}.
    \end{align}
\end{thm}
    
\begin{proof}
    Let $t_{1}\geq T_{f}+d_{f}$ and $x\coloneqq f(t_{1})$. 
    Moreover, we define $t_{0}\coloneqq \lpi{f}(x)=\inf\left\{ t\mid f(t)\ge x\right\} =\inf\left\{ t\mid f(t)\ge f(t_{1})\right\} $.
    By definition, it is clear that $t_{0}\leq t_{1}$ ($t_{1}$ satisfies the condition in the infimum, and $t_{0}$ is the smallest to satisfy it). 
    Since $f$ is non-UC and we have by definition $t_{1}-d_{f}\geq T_{f},$ it follows that 
    \begin{equation*}
        f(T_{f})\leq f(t_{1}-d_{f}) \overset{\eqref{eq:UPP-property}}{=}f(t_{1})-c_{f} < f(t_{1})=f(t_{0}),
    \end{equation*}
    where we used in the strict inequality that $f$ is not UC and thus $t_{0} > T_{f}$. 
    Therefore, 
    \begin{align*}
        \lpi{f}\left(x+k\cdot d_{\lpi{f}}\right) & =\inf\left\{ t\mid f(t)\ge x+k\cdot d_{\lpi{f}}\right\} \\
        & \overset{\eqref{eq:lpi-d}}{=}\inf\left\{ t\mid f(t)\ge x+k\cdot c_{f}\right\} \\
        & =\inf\left\{ t\mid f(t)\ge f(t_{1})+k\cdot c_{f}\right\} \\
        & =\inf\left\{ t\mid f(t)\ge f(t_{0})+k\cdot c_{f}\right\} \\
        & =\inf\left\{ t\mid f(t)\ge f(t_{0}+k\cdot d_{f})\right\} \\
        & =t_{0}+k\cdot d_{f}\\
        & \overset{\eqref{eq:lpi-c}}{=}\lpi{f}(x)+k\cdot c_{\lpi{f}}.
    \end{align*}
\end{proof}

It follows from \Cref{th:UppLowerPseudoInverse} that, in order to compute a representation $R_{\lpi{f}}$, we need only to compute $S^{D'}_{\lpi{f}}$ where
\begin{equation*}
    D' = [0, T_{\lpi{f}} + d_{\lpi{f}}[ = [0, f(T_f + d_f) + c_f[ .
\end{equation*}
If there is no left-discontinuity in $T_f + 2 \cdot d_f$,
it follows that 
\begin{equation*}
    S^{D'}_{\lpi{f}} = \lpib{S^D_f},   
\end{equation*}
where
\begin{equation}
    \label{eq:UppLowerPseudoInverseD}
    D = \left[ 0, T_f + 2 \cdot d_f \right[
\end{equation}
Otherwise, let $x_1 = f \left( ( T_f + 2 \cdot d_f)^- \right)$ and $x_2 = f \left( T_f + 2 \cdot d_f \right)$, then $x_1 < x_2$, and therefore $S^{D'}_{\lpi{f}}$ must end with a constant segment defined in $]x_1, x_2[$ with ordinate $T_f + 2 \cdot d_f$.
Such segment must be added manually at the end of $\lpib{S^D_f}$. 

A similar result can be derived for the upper pseudo-inverse.

\begin{thm}
\label{th:UppUpperPseudoInverse}
    Let $f$ be a non-decreasing UPP function that is not UC. 
    Then, the upper pseudo-inverse $\upi{f}(x)=\sup\left\{ t\mid f(t)\leq x\right\} $ is again UPP with
    \begin{align}
        T_{\upi{f}} & =f\left(T_{f}\right)\\
        d_{\upi{f}} & =c_{f}\label{eq:upi-d}\\
        c_{\upi{f}} & =d_{f}\label{eq:upi-c}
    \end{align}
\end{thm}
    
\begin{proof}
    The proof follows the same steps as the one for the lower pseudo-inverse. 
    Let $t_{0}\geq T_{f}$ and $x\coloneqq f(t_{0})$. 
    Moreover, we define $t_{1}\coloneqq \upi{f}(x)=\sup\left\{ t\mid f(t)\leq x\right\} =\sup\left\{ t\mid f(t)\leq f(t_{0})\right\} $.
    By definition, it is clear that $t_{0}\leq t_{1}$ ($t_{0}$ satisfies the condition in the supremum, and $t_{1}$ is the largest to satisfy it). 
    Since $f$ is non-UC and we have by definition $t_{0}\geq T_{f},$ it follows that 
    \begin{equation*}
        f(t_{0}+d_{f}) \overset{\eqref{eq:UPP-property}}{=}f(t_{0})+c_{f} > f(t_{0})=f(t_{1}).
    \end{equation*}
    where we used in the strict inequality that $f$ is not ultimately constant and thus $t_{1}<t_{0}+d_{f}<\infty$. 
    Therefore, 
    \begin{align*}
        \upi{f}\left(x+k\cdot d_{\upi{f}}\right) & =\sup\left\{ t\mid f(t)\leq x+k\cdot d_{\upi{f}}\right\} \\
        & \overset{\eqref{eq:upi-d}}{=}\sup\left\{ t\mid f(t)\leq x+k\cdot c_{f}\right\} \\
        & =\sup\left\{ t\mid f(t)\leq f(t_{0})+k\cdot c_{f}\right\} \\
        & =\sup\left\{ t\mid f(t)\leq f(t_{1})+k\cdot c_{f}\right\} \\
        & =\sup\left\{ t\mid f(t)\leq f(t_{1}+k\cdot d_{f})\right\} \\
        & =t_{1}+k\cdot d_{f}\\
        & \overset{\eqref{eq:upi-c}}{=}\upi{f}(x)+k\cdot c_{\upi{f}}.
    \end{align*}
\end{proof}

Similar to the previous theorem, it follows from \Cref{th:UppUpperPseudoInverse} that, in order to compute a representation $R_{\upi{f}}$, we need only to compute $S^{D'}_{\upi{f}}$, where 
\begin{equation*}
    D' = [0, T_{\upi{f}} + d_{\upi{f}}[ = [0, f(T_f) + c_f[.
\end{equation*}
If there is no left-discontinuity in $T_f + d_f$,
it follows that 
\begin{equation*}
    S^{D'}_{\upi{f}} = \upib{S^D_f},
\end{equation*}
where 
\begin{equation*}
    D = [0, T_f + d_f[.
\end{equation*}
Otherwise, let $x_1 = f \left( ( T_f + d_f)^- \right)$ and $x_2 = f \left( T_f + d_f \right)$, then $x_1 < x_2$, and therefore $S^{D'}_{\upi{f}}$ must end with a constant segment defined in $]x_1, x_2[$ with ordinate $T_f + d_f$.
Such segment must be added manually at the end of $\upib{S^D_f}$. 
The alert reader will notice that $T_{\lpi{f}}$ and $T_{\upi{f}}$ differ, for which we can provide the following intuitive explanation.
Consider a function $f$ so that $f(t) = k, \forall t \in ]a, T + b[$ with $a < T, b > 0$. 
Then $\lpi{f}(k) = a$, as the lower pseudo-inverse ``goes backwards'' to the start of the constant segment.
However, since $a < T$, the pseudo-periodic property does not apply for $f(a)$, i.e., we cannot say anything about $f(a + d)$.
So, in general, we cannot say $\lpi{f}$ is pseudo-periodic from $f(T)$, and we instead need to ``skip'' to the second pseudo-period so that, as in the proof, $T < a < T + d$. \\
The same does not apply for $\upi{f}$, however, since $\upi{f}(k) = T + b$ as the upper pseudo-inverse ``goes forward'' to the end of the constant segment and $T + b > T$, thus we can rely on the pseudo-periodic property of $f$. 

An interesting consequence of this discussion is that the representation $R_f$ may change when we do not expect it to.
From \cite[pp.~64]{Liebeherr2017}, \cite[p.~48]{Bouillard2018DNC}, we know the following properties:
\begin{itemize}
    \item if $f$ is left-continuous, $\lpib{\lpi{f}} = \lpib{\upi{f}} = f$,
    \item if $f$ is right-continuous, $\upib{\upi{f}} = \upib{\lpi{f}} = f$.
\end{itemize}
Thus, one may expect that applying the pseudo-inverse twice would lead to a function with the \emph{same} representation, i.e., that $\lpib{\lpib{R_f}} = R_f$.
However, as per the discussion above, the start of the pseudo-period of the result would go from an initial $T_f$ to $T_f + d_f + c_f$. 
This is unavoidable -- the above example shows that there exists one case when $T_f$ would \emph{not} be the correct starting point. 
However, in other cases, $T_f$ would be the correct starting point for the pseudo-periodic behavior.

This is an instance of a general issue encountered with algorithms for UPP curves -- also discussed in \cite{ZippoEtAlii2022}. Generic algorithms, that make no assumptions on the shape of the operands (such as the ones presented here for the pseudo-inverse), may in general yield \emph{non-minimal representations} of the result. 
Generally speaking, minimal representations are preferable, since the number of elements in a sequence affects the complexity of the algorithms. 
However, addressing the issue of representation minimization \emph{a priori} when implementing NC operators is too hard (if doable at all), since one would need to make a comprehensive list of subcases, and, of course, as many formal proofs of correctness.
It is instead considerably more efficient to devise a generic algorithm for an operator, neglecting minimization, and then use a simple algorithm \emph{a posteriori} that minimizes the representation of the result -- see \cite{ZippoEtAlii2022}.

\subsection{By-sequence algorithm for pseudo-inverses}
\label{subsec:PseudoInverses-Algorithm}

In this section we discuss the by-sequence algorithms for pseudo-inverses.
We recall that with "by-sequence" we mean that the operand, and thus its result, is defined on a limited domain.
Without loss of generality, we will focus on a sequence $S$ representing a function $f$ over an interval $\left[0, t\right[$, with $f(0) = 0$.
Then $\lpi{S}$ is the sequence representing $\lpi{f}$ over the interval $\left[0, f(t^-)\right[$. The same applies to $\upi{S}$.


The simplest case is when $S$ is continuous and strictly increasing, hence bijective. 
In this case, both $\lpi{S}$ and $\upi{S}$ are the classic \emph{inverse} of $S$, and the algorithm consists of drawing, for each point and segment of $S$, its reflection over the line $y = x$.
However, when $S$ includes discontinuities and/or constant segments, the algorithm becomes slighlty more complicated: a discontinuity in $S$ ``maps'' to a constant segment in both $\lpi{S}$ and $\upi{S}$, while a constant segment in $S$ ``maps'' to a right-discontinuity in $\lpi{S}$ and a left-discontinuity in $\upi{S}$.
This is exemplified in \Cref{fig:lpi-sequence-example}.

\begin{figure}[!t]
    \centering
    \begin{subfigure}{0.45\linewidth}
        \includegraphics[width=\linewidth]{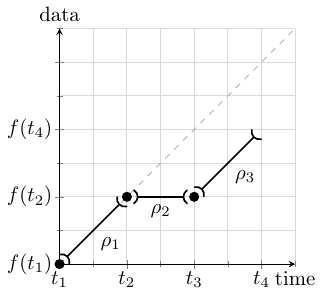}
        \caption{$S$}
        \label{fig:lpi-sequence-example.f}
    \end{subfigure}
    \begin{subfigure}{0.45\linewidth}
        \includegraphics[width=\linewidth]{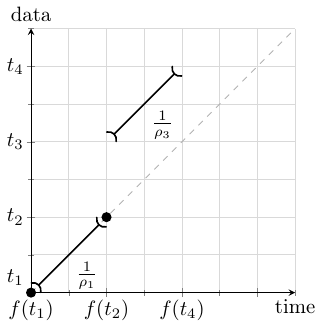}
        \caption{$\lpi{S}$}
        \label{fig:lpi-sequence-example.lpi}
    \end{subfigure}
    \begin{subfigure}{0.45\linewidth}
        \includegraphics[width=\linewidth]{tikzcache/lpi-sequence-example-f}
        \caption{$\lpib{\lpi{S}}$}
        \label{fig:lpi-sequence-example.lpi-of-lpi}
    \end{subfigure}
    \caption{Example of lower pseudo-inverse of a sequence $S$. 
        Since $S$ is left-continuous, $S = \lpib{\lpi{S}}$}
    \label{fig:lpi-sequence-example}
\end{figure}

We describe \Cref{alg:lower-pseudo-code-seq} for the lower pseudo-inverse (the one for the upper pseudo-inverse differs in few details which we briefly discuss later). 
\Cref{alg:lower-pseudo-code-seq} linearly scans $S$ considering one elements at a time. 
Based on the type of element (point or segment), as well as on its topological relationship with its predecessor, it decides what to add to $\lpi{S}$. 

More in detail, there are eight possible cases, shown in \Cref{table:cases}, which require zero, one, two, or three elements to be added to $\lpi{S}$. 
These are reported in the same order in \Cref{alg:lower-pseudo-code-seq}. 
The rigorous (though cumbersome) mathematical justification for each case is instead postponed to \Cref{app-sec:PseudoInverses-Formalization} for the benefit of the interested reader.

\begin{table}[ht]
    \caption{Cases to be considered in the by-sequence algorithm to compute $\lpi{S}$}
    \label{table:cases}
    \centering
    \small
    \begin{tabular}{| p{0.15\linewidth} | p{0.15\linewidth} | p{0.15\linewidth} | p{0.15\linewidth} | p{0.15\linewidth} | p{0.15\linewidth} |}
        \hline
        \thead{Considered\\ Element} & \thead{Constant\\ segment} & \thead{Discontinuity\\ in $S$} & \thead{Example of $S$\\\\} & \thead{Append to \\ $\lpi{S}$} & \thead{Case \#} \\
        \hline
        \multirow{4}{*}{\makecell{\textbf{Point}\\ after segment}} & \multirow{2}{*}{Yes} & 
        Yes &
        \includegraphics[width=\linewidth]{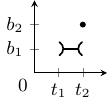} & 
        \includegraphics[width=\linewidth]{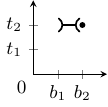} & 
        c1 \\
        \cline{3-6}
        & &
        No &
        \includegraphics[width=\linewidth]{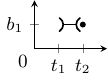} & 
        \makecell{\emph{nothing}\\\emph{to append}} & 
        c2 \\
        \cline{2-6}
        & \multirow{2}{*}{No} & 
        Yes &
        \includegraphics[width=\linewidth]{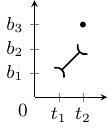} & 
        \includegraphics[width=\linewidth]{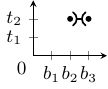} & 
        c3 \\
        \cline{3-6}
        & &
        No &
        \includegraphics[width=\linewidth]{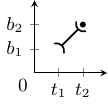} & 
        \includegraphics[width=\linewidth]{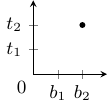} & 
        c4 \\
        \hline
        \multirow{4}{*}{\makecell{\textbf{Segment}\\ after point}} & \multirow{2}{*}{Yes} & 
        Yes &
        \includegraphics[width=\linewidth]{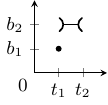} & 
        \includegraphics[width=\linewidth]{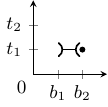} & 
        c5 \\
        \cline{3-6}
        & &
        No &
        \includegraphics[width=\linewidth]{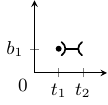} & 
        \makecell{\emph{nothing}\\\emph{to append}} & 
        c6 \\
        \cline{2-6}
        & \multirow{2}{*}{No} & 
        Yes &
        \includegraphics[width=\linewidth]{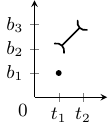} & 
        \includegraphics[width=\linewidth]{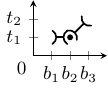} & 
        c7 \\
        \cline{3-6}
        & &
        No &
        \includegraphics[width=\linewidth]{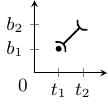} & 
        \includegraphics[width=\linewidth]{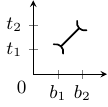} & 
        c8 \\
        \hline
    \end{tabular}
\end{table}

\begin{algorithm}[ht]
    \caption{Pseudo-code for lower pseudo-inverse of a finite sequence} 
    \label{alg:lower-pseudo-code-seq}
    \hspace*{\algorithmicindent} \textbf{Input} Finite sequence of elements $ S $, consisting of $ e_k, k \in \left\{1, \dots, n\right\} $ that is either a point or a segment.
    Moreover, $ e_1 $ is a point at the origin $ (0, 0). $\\
    \hspace*{\algorithmicindent} \textbf{Return} Lower pseudo-inverse $ \lpi{S} $ of $ S $, consisting of a sequence of elements $ O = \left\{o_1, \dots, o_m\right\} $
    \begin{algorithmic}[1]
        \State Define an empty sequence of elements $ O \coloneqq \{\,\} $
        \State Append $ p \coloneqq (0,0) $ to $ O $ \Comment{$ \lpi{f}(e_0) $}
        \For{$ e_k $ \textbf{in} $ \left(e_2, \dots, e_n\right) $}
        \Comment{The for loop starts after the origin}
            \If{$ e_k == p_i $} \Comment{The element is a point}
            \label{line:lpi-sp}
                \State $ e_{k-1} $ is a segment $ s_{i-1} $
                \If{$ s_{i-1} $ is constant}
                \label{line:lpi-sp-const}
                    \If{$ f(t_i^-) < f(t_i) $} \Comment{$f$ is not left-cont. at $ t_i $}
                    \label{line:lpi-sp-const-jump}
                    \State Append $ s \coloneqq \left(f(t_{i}^{-}), f(t_{i}), t_i, t_i\right) $ to $ O $ \Comment{(c1)}
                    \State Append $ p \coloneqq \left(f(t_i), t_i\right) $ to $ O $
                \Else \Comment{$f$ is left-cont. at $ t_i $}
                \label{line:lpi-sp-const-no-jump}
                    \State Nothing to append	\Comment{(c2)}
                \EndIf
            \Else \Comment{$ s_{i-1} $ is not constant}
            \label{line:lpi-sp-inc}
                \If{$ f(t_i^-) < f(t_i) $} \Comment{$f$ is not left-cont. at $ t_i $}
                \label{line:lpi-sp-inc-jump}
                    \State Append $ p \coloneqq \left(f(t_i^-), t_i\right) $ to $ O $ \Comment{(c3)}
                    \State Append $ s \coloneqq \left(f(t_{i}^{-}), f(t_{i}), t_i, t_i\right) $ to $ O $
                    \State Append $ p \coloneqq \left(f(t_i), t_i\right) $ to $ O $
                \Else \Comment{$f$ is left-cont. at $ t_i $}
                \label{line:lpi-sp-inc-no-jump}
                    \State Append $ p \coloneqq \left(f(t_i), t_i\right) $ to $ O $ \Comment{(c4)}
                \EndIf
            \EndIf
        \Else \Comment{The element is a segment $ s_i $}
        \label{line:lpi-ps}
        \State $ e_{k-1} $ is a point $ p_{i} $
            \If{$ e_k = s_{i} $ is constant}
            \label{line:lpi-ps-const}
                \If{$ f(t_i) < f(t_i^+) $} \Comment{$f$ is not right-cont. at $ t_i $}
                \label{line:lpi-ps-const-jump}
                    \State Append $ s \coloneqq \left(f(t_{i}), f(t_{i}^+), t_i, t_i\right) $ to $ O $ \Comment{(c5)}
                    \State Append $ p \coloneqq \left(f(t_i^+), t_i\right) $ to $ O $
                \Else \Comment{$f$ is right-cont. at $ t_i $}
                \label{line:lpi-ps-const-no-jump}
                    \State Nothing to append \Comment{(c6)}
                \EndIf
            \Else \Comment{$ s_{i} $ is not constant}
            \label{line:lpi-ps-inc}
                \If{$ f(t_i) < f(t_i^+) $} \Comment{$f$ is not right-cont. at $ t_i $}
                \label{line:lpi-ps-inc-jump}
                    \State Append $ s \coloneqq \left(f(t_{i}), f(t_{i}^+), t_i, t_i\right) $ to $ O $ \Comment{(c7)}
                    \State Append $ p \coloneqq \left(f(t_i^+), t_i\right) $ to $ O $
                    \State Append $ s \coloneqq \left(f(t_i^+), f(t_{i+1}^-), t_i, t_{i+1}\right) $ to $ O $
                \Else \Comment{$f$ is right-cont. at $ t_i $}
                \label{line:lpi-ps-inc-no-jump}
                    \State Append $ s \coloneqq \left(f(t_i^+), f(t_{i+1}^-), t_i, t_{i+1}\right) $ to $ O $ \Comment{(c8)}
                                        \EndIf
                \EndIf
            \EndIf
        \EndFor
    \end{algorithmic}
\end{algorithm}

We exemplify the above algorithm with reference to the example in \Cref{fig:lpi-sequence-example}.
For each of the considered steps, we will reference the case in \Cref{table:cases}, the line of \Cref{alg:lower-pseudo-code-seq}, and the relevant equations from \Cref{app-sec:PseudoInverses-Formalization} proving the result.
Processing each element from left to right, we calculate:
\begin{itemize}
    \item The origin $ (t_1, f(t_1)) = (0, 0) $ for $ \lpi{f}(0) $
    \item for the segment $ s_1 $ and its predecessor point $ p_1 = (t_1, f(t_1)) $:
    this corresponds to \Cref{line:lpi-ps} of the algorithm.
    Since $ s_1 $ has a positive slope, we continue in \Cref{line:lpi-ps-inc}.
    As the function is right-continuous at $ t_1 $, we are in case c8.
    We go to \Cref{line:lpi-ps-inc-no-jump} and add a segment $ s = \left(f(t_1^+), f(t_{2}^-), t_1, t_{2}\right) $ to $ O $.
    It can be verified that this follows \Cref{eq:lpi-ps-inc-no-jump}.

    \item for the point $p_2 = \left(t_2, f(t_2)\right)$ and its preceding segment $ s_1, $ we are in case c4, corresponding to \Cref{line:lpi-sp-inc-no-jump} and we therefore append the point $ p \coloneqq \left(f(t_2), t_2\right) $ to $ O $.
    It can be verified that this follows \Cref{eq:lpi-sp-inc-no-jump}.

    \item for the constant segment $ s_2 $ with preceding point $p_2 =  \left(t_2, f(t_2)\right)$, we are in case c6, corresponding to \Cref{line:lpi-ps-const-no-jump}, and no element is added.
    This follows \Cref{eq:lpi-ps-const-no-jump}. 

    \item for the point $ p_3 = \left(t_3, f(t_3)\right) $ with preceding constant segment $ s_2 $, we are in case c2, corresponding to \Cref{line:lpi-sp-const-no-jump}, and no element is added.
    This follows \Cref{eq:lpi-sp-const-no-jump}.

    \item for a segment $ s_3 $ with preceding point $\left(t_3, f(t_3)\right)$, we are in case c8, \Cref{line:lpi-ps-inc-no-jump}, and append $ s \coloneqq \left(f(t_3^+), f(t_{4}^-), t_3, t_{4}\right) $ to $ O $ (verified in \Cref{eq:lpi-ps-inc-no-jump}).
\end{itemize}


We note that, since $\lpi{S}$ is left-continuous, when a continuous sequence of a point, a constant segment, and a point is encountered in $S$, they all ``map'' to the inverse of the first (leftmost) point of this sequence.
This justifies the fact that nothing has to be added to $\lpi{S}$ in these cases (e.g., 2 and 6).

The algorithm for $\upi{S}$, that we omit here for brevity, differs from the one provided only in how constant segments are handled, that is, by appending the inverse of the \emph{last} (rightmost) point instead of the first (recall that the upper pseudo-inverse is right-continuous). 
This requires the algorithm for $\upi{S}$ to look ahead to the next element during the linear scan. 
We leave the (tedious, but simple) task of spelling out the minutiae of this algorithm to the interested reader. 

\FloatBarrier

\subsection{By-curve algorithm for pseudo-inverses}
\label{subsec:PseudoInverses-Conclusions}

We can now discuss the by-curve algorithm by combining the results presented in \Cref{subsec:PseudoInverses-Theorems,subsec:PseudoInverses-Algorithm}.
In \Cref{alg:lower-pseudo-code-curve}, we show the pseudo-code to compute $\lpi{f}$ for a UPP function $f$.
The analogous for upper pseudo-inverse, which we omit for brevity, can be similarly derived from the results in the sections above.

\begin{algorithm}[ht]
    \caption{Pseudo-code for lower pseudo-inverse of a UPP function}
    \label{alg:lower-pseudo-code-curve}
    \hspace*{\algorithmicindent} \textbf{Input} Representation $R_f$ of a non-decreasing UPP function $f$, consisting of sequence $S_f$ and parameters $T_f$, $d_f$ and $c_f$. \\
    \hspace*{\algorithmicindent} \textbf{Return} Representation $R_{\lpi{f}}$ of $\lpi{f}$
    \begin{algorithmic}[1]
        \State Compute the UPP parameters for the result \Comment{\Cref{th:UppLowerPseudoInverse}}
            \State\hspace{\algorithmicindent} $T_{\lpi{f}} \gets f\left(T_{f}+d_{f}\right)$
            \State\hspace{\algorithmicindent} $d_{\lpi{f}} \gets c_{f}$
            \State\hspace{\algorithmicindent} $c_{\lpi{f}} \gets d_{f}$
        \State Compute $S^D_f$ \Comment{\Cref{eq:UppLowerPseudoInverseD}}
            \State\hspace{\algorithmicindent} $D \gets [0, T_f + 2 \cdot d_f[$
            \State\hspace{\algorithmicindent} $S^D_f \gets \text{Cut}(R_f, D)$
        \State Compute $S_{\lpi{f}} \gets \lpib{S^D_f}$ \Comment{\Cref{alg:lower-pseudo-code-seq}}
        \State $R_{\lpi{f}} \gets \left(S_{\lpi{f}}, T_{\lpi{f}}, d_{\lpi{f}}, c_{\lpi{f}}\right)$
    \end{algorithmic}
\end{algorithm}

Regarding the complexity of \Cref{alg:lower-pseudo-code-curve}, we note that the main cost is computing $\lpib{S^D_f}$.
Since \Cref{alg:lower-pseudo-code-seq} is a linear scan of the input sequence, the resulting complexity is $\mathcal{O}\left(\card{S^D_f}\right)$.

\subsection{Corner cases: UC and wUI functions}
\label{subsec:PseudoInverses-CornerCases}

We conclude this section by discussing the two corner cases that we had initially left out, i.e., those when $f$ is either UC or wUI.
For these, some mathematical inconsistencies need be resolved first. For example:
\begin{itemize}
    \item if $f$ was UC, the current algorithm would yield $d_{\lpi{f}} = 0$,
    \item if $f$ was wUI, it would yield $T_{\lpi{f}} = +\infty$.
\end{itemize}

We treat these two cases in the following propositions.

\begin{prop}
    Let $ f:\mathbb{Q}^+ \to \mathbb{Q}^+ $ be a non-decreasing, UC function with $ T_f \in \mathbb{Q}^+. $
    Then, its lower pseudo-inverse $ \lpi{f}(y) $ is for $ y \geq f(T_f) $
    \begin{equation*}
        \lpi{f}(y) = \begin{cases}
            \inf\left\{ x\mid f(x)\geq \overbrace{y}^{=f(T_f)}\right\} = T_f, & \text{if }  y = f(T_f),\\
            \sup\left\{ x\mid f(x) < y\right\} =  +\infty, & \text{if }  y > f(T_f)
        \end{cases}
    \end{equation*}
    and its upper pseudo-inverse $ \upi{f}(y) $ is for $ y \geq f(T_f) $
    \begin{equation*}
        \upi{f}(y) = \sup\left\{ x\mid f(x) \leq y\right\} = +\infty.
    \end{equation*}
\end{prop}

Note that, if one follows \Cref{def:wUI-function}, the result of both operators is \emph{weakly ultimately infinite}, while only $\upi{f}$ is \emph{ultimately infinite}.

\begin{prop}
    Let $ f:\mathbb{Q}^+ \to \mathbb{Q}^+ \cup \{+\infty\} $ be a non-decreasing, weakly ultimately infinite function with $ T_f \in \mathbb{Q}^+ $. 
    We recall that this is defined as
    \begin{align*}
        f(t) &< +\infty, \qquad \forall t < T_f,\\
        f(t) &= +\infty, \qquad \forall t > T_f.
    \end{align*}
    Then, its lower pseudo-inverse $ \lpi{f}(y) $ is for $ y \geq f\left(T_f^-\right) $
    \begin{equation*}
        \lpi{f}(y) = \sup\left\{ x\mid f(x) < y\right\} = T_f
    \end{equation*}
    and its upper pseudo-inverse $ \upi{f}(y) $ is for $ y \geq f\left(T_f^-\right) $
    \begin{equation*}
        \upi{f}(y) = \sup\left\{ x\mid f(x) \leq y\right\} = T_f.
    \end{equation*}
    In other words, both pseudo-inverses are UC.
\end{prop}

Again, by using \Cref{def:wUI-function}, we show that the above property applies to \emph{weakly ultimately infinite} functions, not only to \emph{ultimately infinite} ones.

Starting from the above results, one can derive the few modifications to the algorithms described so far in this section to include these two corner cases. 
We leave this simple (yet tedious) task to the interested reader.

\FloatBarrier

\section{Composition of UPP Functions}
\label{sec:Composition}

In this section, we discuss the composition operator for UPP functions, i.e., $(f \circ g)(t)=f(g(t))$.
This section is organized as follows.
In \Cref{subsec:Composition-Theorems} we show that the composition of UPP functions is again UPP, and provide expressions to compute its UPP parameters a priori.
In \Cref{subsec:Composition-Algorithm} we discuss, first via an example and then via pseudo-code, how to compute the composition algorithmically.
In \Cref{subsec:Composition-Conclusions} we conclude with a summary of the by-curve algorithm and some observations on the algorithmic complexity of this operator.

\subsection{Properties of composition of UPP functions}
\label{subsec:Composition-Theorems}

We assume that the inner function $g$ is not wUI.\footnote{
    If, for $t > T$, $g(t) = +\infty$ then $h(t) = \lim_{y \rightarrow +\infty} f(y)$.
    The fact that $f$ is UPP does not guarantee that such limit exists, e.g., when $f$ is periodic. 
} 
We initially provide the result for generic $f$ and $g$. 
Later on, we show that, if either or both are UA or UC, we can improve upon this result. 

\begin{thm}
\label{thm:UPP-composition} 
    Let $f$ and $g$ be two UPP functions with $g$ being non-negative, non-decreasing and not wUI.
    Then, their composition $h \coloneqq f \circ g$ is again UPP with
    \begin{align}
        T_{h} & =\max\left\{ \lpi{g}(T_{f}),T_{g}\right\} ,\label{eq:T_h}\\
        d_{h} & =N_{d_{f}}\cdot d_{g}\cdot D_{c_{g}},\label{eq:d_h-UPP}\\
        c_{h} & =D_{d_{f}}\cdot N_{c_{g}}\cdot c_{f}.\label{eq:c_h-UPP}
    \end{align}
\end{thm}

\begin{proof}
    Let $k_{h}\in\mathbb{N}$ be arbitrary but fixed. Since $g$ is UPP, it holds for all $t\geq T_{g}$ that
    \begin{align*}
        h(t+k_{h}\cdot d_{h}) & =f\left(g(t+k_{h}\cdot d_{h})\right)\\
        & =f\left(g\left(t+k_{h}\cdot\frac{d_{h}}{d_{g}}\cdot d_{g}\right)\right)\\
        & \overset{\eqref{eq:UPP-property}}{=}f\left(g(t)+k_{h}\cdot\frac{d_{h}}{d_{g}}\cdot c_{g}\right),
    \end{align*}
    where we used the UPP property of $g$ in the last line.
    Note that $k_{g}\coloneqq k_{h}\cdot\frac{d_{h}}{d_{g}}\in\mathbb{N}_{0},$ since
    $\frac{d_{h}}{d_{g}}\overset{\eqref{eq:d_h-UPP}}{=}N_{d_{f}}\cdot D_{c_{g}}\in\mathbb{N},$
    where we used that $d_{f}>0$. 
    Moreover, since $f$ is UPP, too, we have under this additional assumption of $g(t)\ge T_{f}$ that 
    \begin{align*}
        h(t+k_{h}\cdot d_{h}) & =f\left(g(t)+k_{h}\cdot\frac{d_{h}}{d_{g}}\cdot c_{g}\right)\\
        & =f\left(g(t)+k_{h}\cdot\frac{d_{h}\cdot c_{g}}{d_{g}\cdot d_{f}}\cdot d_{f}\right)\\
        & \overset{\eqref{eq:UPP-property}}{=}f(g(t))+k_{h}\cdot\frac{d_{h}\cdot c_{g}}{d_{g}\cdot d_{f}}\cdot c_{f}\\
        & =h(t)+k_{h}\cdot\frac{d_{h}\cdot c_{g}\cdot c_{f}}{d_{g}\cdot d_{f}}\\
        & \overset{\eqref{eq:c_h-UPP}}{=}h(t)+k_{h}\cdot c_{h}.
    \end{align*}
    Note that $k_{f}\coloneqq k_{h}\cdot\frac{d_{h}\cdot c_{g}}{d_{g}\cdot d_{f}}\in\mathbb{N}_{0}$,
    since $\frac{d_{h}\cdot c_{g}}{d_{g}\cdot d_{f}}\overset{\eqref{eq:d_h-UPP}}{=}\frac{N_{d_{f}}}{d_{f}}\cdot D_{c_{g}}\cdot c_{g}=D_{d_{f}}\cdot N_{c_{g}}\in\mathbb{N}_{0}$
    and we used that $c_{g}\geq0$ as $g$ is non-decreasing. 

    We set $t\geq T_{g}$ and $g(t)\ge T_{f}$, thus ensuring that both $f$ and $g$ are in their periodic part. Exploiting the notion of a lower pseudo-inverse
    and $g$ being non-decreasing, the latter expression is equivalent
    to $t\geq \lpi{g}(T_{f})$ \cite[p. 62]{Liebeherr2017}. 
    Therefore, we require 
    \begin{equation*}
    t\geq T_{h}\overset{\eqref{eq:T_h}}{=}\max\left\{ \lpi{g}(T_{f}),T_{g}\right\} .
    \end{equation*}
    This concludes the proof.
\end{proof}

\begin{rem}
    Note that the above is also true for the particular case in which $d_{f}\in\mathbb{N},c_{g}\in\mathbb{N}_{0}$. 
    In fact, it follows that $N_{d_{f}}=d_{f}$ and $D_{c_{g}}=1$ and thus
    \begin{equation*}
        d_{h}\overset{\eqref{eq:d_h-UPP}}{=}N_{d_{f}}\cdot D_{c_{g}}\cdot d_{g}=d_{f}\cdot d_{g},
    \end{equation*}
    and the properties are then verified since $\frac{d_{h}}{d_{g}}=d_{f}\in\mathbb{N}$ and $\frac{d_{h}\cdot c_{g}}{d_{g}\cdot d_{f}}=c_{g}\in\mathbb{N}_{0}$.
    The corresponding $c_{h}$ is $c_{f}\cdot c_{g}$.
\end{rem}

It follows from \Cref{thm:UPP-composition} that, in order to compute the representation $R_{h}$, we only need to compute $S^{D_h}_{h}$ where
\begin{equation*}
    D_h = [0, T_h + d_h[ = \left[0, \max\left\{ \lpi{g}(T_{f}),T_{g}\right\} + N_{d_{f}}\cdot d_{g}\cdot D_{c_{g}} \right[.
\end{equation*}
It follows that
\begin{equation*}
    S^{D_h}_h = S^{D_f}_f \circ S^{D_g}_g, 
\end{equation*}
where
\begin{equation}
    \label{eq:CompositionD}
    \begin{aligned}
        D_g &= \left[ 0, T_h + d_h \right[ \\
        D_f &= \left[ g(0), g\left((T_h + d_h)^-\right) \right] 
    \end{aligned}
\end{equation}

The reason $D_f$ needs to be right-closed is that $S^{D_g}_g$ may end with a constant segment. 
If this happens, $\exists \overline{t} \in D_g$ so that $g(\overline{t}) = g\left((T_h + d_h)^-\right)$ thus we will need to compute $f \left( g \left( \overline{t} \right) \right) = f \left( g\left((T_h + d_h)^-\right) \right)$, and that is in fact the right boundary of $D_f$.
On the other hand, if $S^{D_g}_g$ ends with a strictly increasing segment, it is safe to have $D_f$ right-open.

Hereafter, we show that the above result can be improved when either or both functions are UA. 
First, we consider the case for which only $g$ is UA.

\begin{prop}
    \label{prop:UPP-composition-g-UA}
    Let $f$ and $g$ be two UPP functions that are not wUI, with $g$ being non-negative, non-decreasing, UA, with $\rho_{g}>0$ (hence not UC).
    Then, their composition $h\coloneqq f\circ g$ is again UPP with
    \begin{align}
        T_{h} & =\max\left\{ \lpi{g}(T_{f}),T_{g}\right\} , \notag\\
        d_{h} & = \frac{d_f}{\rho_g}, \label{eq:d_h-g-UA}\\
        c_{h} & = c_{f}. \label{eq:c_h-g-UA}
    \end{align}
\end{prop}

\begin{proof}
    Let $k_{h}\in\mathbb{N}$ be arbitrary but fixed. Since $g$ is assumed to be UA, it holds for all $t\geq T_{h}$ that
    \begin{align*}
        h(t+k_{h}\cdot d_{h}) & =f\left(g(t+k_{h}\cdot d_{h})\right)\\
        & \overset{\eqref{eq:ultimately-affine}}{=} f\left(g(T_{h})+\rho_{g}\cdot\left(t+k_{h}\cdot d_{h}-T_{h}\right)\right)\\
        & =f\left(\rho_{g}\cdot t+g(T_{h})-\rho_{g}\cdot T_{h}+k_{h}\cdot d_{f}\right)\\
        & \overset{\eqref{eq:UPP-property}}{=} f\left(\rho_{g}\cdot t+g(T_{h})-\rho_{g}\cdot T_{h}\right)+k_{h}\cdot c_{f}\\
        & =f\left(g(T_{h})-\rho_{g}\cdot\left(t-T_{h}\right)\right)+k_{h}\cdot c_{f}\\
        & \overset{\eqref{eq:ultimately-affine}}{=}f\left(g(t)\right)+k_{h}\cdot c_{f}\\
        & =h(t)+k_{h}\cdot c_{h}.
    \end{align*}
\end{proof}

Again, in order to compute the representation $R_{h},$ we only need $S_{h}^{D_{h}}$, where
\begin{equation*}
    D_{h}=\left[0,T_{h} + d_h\right[=\left[0,\max\left\{ \lpi g(T_{f}),T_{g}\right\} + \frac{d_f}{\rho_g} \right[.
\end{equation*}
It follows that 
\begin{equation*}
    S_{h}^{D_{h}}=S_{f}^{D_{f}}\circ S_{g}^{D_{g}},
\end{equation*}
where
\begin{equation}
    \begin{aligned}
        D_{g} & =\left[0,T_{h}+d_{h}\right[\\
        &\overset{\eqref{eq:d_h-g-UA}}{=} \left[0,T_{h}+\frac{d_{f}}{\rho_{g}}\right[\\
        D_{f} & =\left[g(0),g\left(\left(T_{h}+d_{h}\right)^{-}\right)\right[\\
        & \overset{\eqref{eq:d_h-g-UA}}{=}\left[g(0),g\left(T_{h}+\frac{d_{f}}{\rho_{g}}\right)\right[\\
            & =\left[g(0),g\left(T_{h}\right)+d_{f}\right[.
    \end{aligned}
    \label{eq:CompositionD-g-UA}
\end{equation}
Here, we observe that domain $D_f$ is smaller than the one obtained by applying directly \Cref{thm:UPP-composition}, due to the disappearance of a factor $D_{d_f} \cdot N_{c_g} \geq 1$. 
In fact, with \Cref{thm:UPP-composition} we would have:
\begin{align*}
    D_{g} & =\left[0,T_{h}+d_{h}\right[\\
    & \overset{\eqref{eq:d_h-UPP}}{=} \left[0,T_{h} + N_{d_f} \cdot d_g \cdot D_{c_g}\right[\\
    & =\left[0,T_{h} + D_{d_f} \cdot N_{c_g} \cdot \frac{d_f}{\rho_g}\right[\\
    D_{f} & =\left[g(0),g\left(\left(T_{h}+d_{h}\right)^{-}\right)\right[\\
    & \overset{\eqref{eq:d_h-UPP}}{=}\left[g(0),g\left(\left(T_{h}+N_{d_f} \cdot d_g \cdot D_{c_g}\right)^{-}\right)\right[\\
    & =\left[g(0),g(T_h) + D_{d_f} \cdot N_{c_g} \cdot d_f\right[.
\end{align*}

As specified in the statement of \Cref{prop:UPP-composition-g-UA}, we exclude the case for which $g$ is UC.
This is because of \Cref{eq:d_h-g-UA} where $\rho_g$ is in the denominator, hence cannot be zero.
However, if $g$ is UC, a stronger proposition can be found as reported in \Cref{app-sec:comp-UCF}.

Next, we consider the case in which only $f$ is UA.

\begin{prop}
\label{prop:UPP-composition-f-UA}
    Let $f$ be UA and $g$ be an UPP function that is non-negative, non-decreasing and not wUI. 
    Then, their composition $h\coloneqq f\circ g$ is again UPP with
    \begin{align}
        T_{h} & =\max\left\{ \lpi{g}(T_{f}),T_{g}\right\} \notag\\
        d_{h} & =d_{g},\label{eq:d_h-f-UA}\\
        c_{h} & = c_{g} \cdot \rho_f.\label{eq:c_h-f-UA}
    \end{align}
\end{prop}

\begin{proof}
    Let $k_{h}\in\mathbb{N}$ be arbitrary but fixed. Since $f$ is assumed to be UA, it holds for all $t\geq T_{h}$ that
    \begin{align*}
        h(t+k_{h}\cdot d_{h}) & =f\left(g(t+k_{h}\cdot d_{h})\right)\\
        & \overset{\eqref{eq:UPP-property}}{=} f\left(g(t) + k_{h}\cdot c_{g}\right)\\
        & \overset{\eqref{eq:ultimately-affine}}{=} f(g(T_h)) + \rho_f \cdot \left(g(t) + k_{h}\cdot c_{g} - g(T_h)\right)\\
        & =f(g(T_h)) + \rho_f \cdot \left(g(t) - g(T_h)\right) + k_{h} \cdot c_{g} \cdot \rho_f\\
        & \overset{\eqref{eq:ultimately-affine}}{=} f(g(t)) + k_{h} \cdot c_{g} \cdot \rho_f\\
        & =h(t) + k_{h} \cdot c_{h}.
    \end{align*}
\end{proof}

Again, for the representation $R_{h},$ we only compute $S_{h}^{D_{h}}$, where
\begin{equation*}
    D_{h}=\left[0,T_{h} + d_h\right[=\left[0,\max\left\{ \lpi g(T_{f}),T_{g}\right\} + d_g \right[.
\end{equation*}
It follows that 
\begin{equation*}
    S_{h}^{D_{h}}=S_{f}^{D_{f}}\circ S_{g}^{D_{g}},
\end{equation*}
where
\begin{equation}
    \begin{aligned}
        D_{g} & =\left[0,T_{h}+d_{h}\right[\\
        &\overset{\eqref{eq:d_h-f-UA}}{=} \left[0,T_{h}+d_g\right[\\
        D_{f} & =\left[g(0),g\left(\left(T_{h}+d_{h}\right)^{-}\right)\right]\\
        & \overset{\eqref{eq:d_h-f-UA}}{=} \left[g(0),g\left(\left(T_{h}+d_g\right)^{-}\right)\right].
    \end{aligned}
    \label{eq:CompositionD-f-UA}
\end{equation}
Again, domain $D_g$ is smaller than the one that \Cref{thm:UPP-composition} would yield, due to the disappearance of a factor $N_{d_f} \cdot D_{c_g} \geq 1$. 
For comparison, \Cref{thm:UPP-composition} yields
\begin{align*}
    D_{g} & =\left[0,T_{h}+d_{h}\right[\\
    & \overset{\eqref{eq:d_h-UPP}}{=} \left[0,T_{h} + N_{d_f} \cdot D_{c_g} \cdot d_g\right[\\
    D_{f} & =\left[g(0),g\left(\left(T_{h}+d_{h}\right)^{-}\right)\right]\\
    & \overset{\eqref{eq:d_h-UPP}}{=} \left[g(0),g\left(\left(T_{h}+N_{d_f} \cdot d_g \cdot D_{c_g}\right)^{-}\right)\right].
\end{align*}



When both functions are UA, we obtain a stronger result by showing that the composition is UA again.

\begin{prop}
\label{prop:UA-composition} 
    Let $f$ and $g$ be UA functions with $g$ being non-negative, non-decreasing and not wUI. 
    Then, their composition $h\coloneqq f\circ g$ is again UA with
    \begin{align}
        T_{h}^{a} & =\max\left\{ \lpi{g}(T_{f}^{a}),T_{g}^{a}\right\} , \notag\\
        \rho_{h} & =\rho_{f}\cdot\rho_{g}.
        \label{eq:rho_h-UA-UA}
    \end{align}
\end{prop}

\begin{proof}
    If $f$ is ultimately $-\infty$ or $+\infty$ or if $g$ is ultimately $+\infty$, the result is trivial. 
    Let us assume that $f$ and $g$ are not wUI. 
    Define $T_{h}^{a}\coloneqq\max\left\{ \lpi{g}(T_{f}^{a}),T_{g}^{a}\right\} $.
    Then we have that for any $t \ge T_{h}^{a}$
    \begin{align*}
        h(t+T_{h}^{a}) & =f\left( g(t+T_{h}^{a}) \right)\\
        & \overset{\eqref{eq:ultimately-affine}}{=}f\left( g(T_{h}^{a}) + \rho_{g} \cdot (t - T_{h}^{a}) \right)\\
        & \overset{\eqref{eq:ultimately-affine}}{=}f(g(T_{h}^{a})) + \rho_{f} \cdot \left( \left(g(T_{h}^{a}) + \rho_{g} \cdot (t - T_{h}^{a})\right) - g(T_{h}^{a}) \right)\\
        & = f(g(T_{h}^{a})) + \rho_f \cdot \rho_g \cdot (t - T_{h}^{a}) \\
        & = h(T_{h}^{a}) + \rho_{f} \cdot \rho_{g} \cdot (t - T_{h}^{a}).
    \end{align*}
\end{proof}

Considering \Cref{eq:CompositionD}, we observe how taking these results into account will yield tighter $D_f, D_g$ than what we obtain with \Cref{thm:UPP-composition}.

Finally, we mention that, if either or both $f$ and $g$ are UC, then the composition can be simplifed further, even with respect to the above properties. 
We report the results in \Cref{app-sec:comp-UCF}.

\subsection{By-sequence algorithm for composition}
\label{subsec:Composition-Algorithm}

In this section, we discuss the by-sequence algorithm for the composition.
Without loss of generality, we focus on sequences $S_g$, representing a non-negative and non-decreasing function $g$ over an interval $\left[0, t\right[$, and $S_f$, representing a function $f$ defined on over the interval $\left[g(0), g(t^-)\right]$.\footnote{We consider $D_f$ to be always right-closed since it yields the correct result for both cases discussed in the previous section. 
The right boundary of $D_f$ is never used as a breakpoint in the algorithm anyway, as imposed by the condition $y_m < g(b^-)$ discussed below.}
Then, $S_h = S_f \circ S_g$ is the sequence representing $h = f \circ g$ over the interval $\left[0, t\right[$.
We use the example shown in \Cref{fig:comp-sequence-example}, where $t = 6$ and $g(t^-) = 4$.

\begin{figure}[!t]
    \centering
    \begin{subfigure}{0.45\linewidth}
        \includegraphics[width=\linewidth]{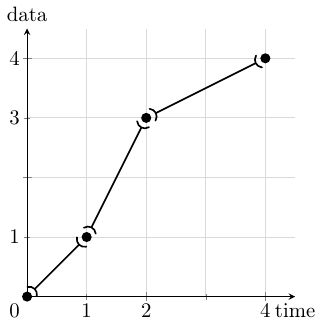}
        \caption{$S_f$}
        \label{fig:comp-sequence-example.f}
    \end{subfigure}
    \begin{subfigure}{0.45\linewidth}
        \includegraphics[width=\linewidth]{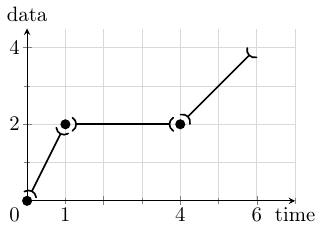}
        \caption{$S_g$}
        \label{fig:comp-sequence-example.g}
    \end{subfigure}
    \begin{subfigure}{0.45\linewidth}
        \includegraphics[width=\linewidth]{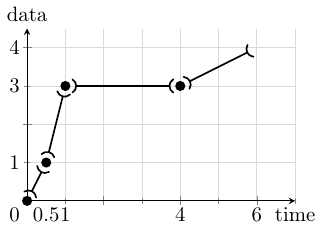}
        \caption{$S_f \circ S_g$}
        \label{fig:comp-sequence-example.h}
    \end{subfigure}
    \caption{Example of composition of two sequences}
    \label{fig:comp-sequence-example}
\end{figure}

First, we consider the shape of $f \circ g$ on an interval $\left]a,b\right[ \subset \left[0, t\right[$, $a,b\in\mathbb{Q}^{+}$. 
Consider the case in which, for this interval, there exist $\rho_{g},\rho_{f}\in\mathbb{Q}^{+}$ so that
\begin{equation}
    \begin{aligned}g(x) & =g(a^{+})+\rho_{g} \cdot (x-a), \qquad\forall x\in\left]a,b\right[,\\
        f(x) & =f\left(g(a^{+}){}^{+}\right)+\rho_{f} \cdot \left(x-g(a^{+})\right), \qquad\forall x\in\left]g(a^{+}),g(b^{-})\right[,
    \end{aligned}
    \label{eq:segment-comp}
\end{equation}
where we use the shorthand notation 
\begin{equation*}
    f\left(g(a^{+}){}^{+}\right) = \lim_{x\to a^{+}}f\left(g(x)\right) = \lim_{y\to y_{0}}f(y),
\end{equation*}
with $y_{0}\coloneqq\lim_{x\to a^{+}}g(x)$.

More broadly speaking, we have segment of $g$ mapping to a segment of $f$.
In the example of \Cref{fig:comp-sequence-example}, $]4, 6[$ is such an interval.  
Then, in this interval we can apply the chain rule and find that $h'(x)=f'(g(x))\cdot g'(x)=\rho_{g}\cdot \rho_{f}$ for all $x\in\left]a,b\right[$. 
Thus, $h$ also a segment on $\left]a,b\right[$.

If either of the equations in (\ref{eq:segment-comp}) does not apply, it means that one function has one or more breakpoints over this interval.
Assume initially that this be $g$. 
Let this finite sets of breakpoints be $t_{0},\dots,t_{n}$, with $a < t_{0} < \cdots < t_{n} < b$. 
Then, the intervals $\left]a,t_{0}\right[,\dots,\left]t_{n},b\right[$ verify the properties in Equation~(\ref{eq:segment-comp}) while for any breakpoint $t_{i}$ we can just compute $f(g(t_{i}))$. 
A similar reasoning can be done for $f$: consider the finite set of breakpoints $y_0,\dots,y_m$, with $g(a^+) < y_{0} < \cdots < y_{m} < g(b^-)$. 
Then, we can use the lower pseudo-inverse of $g$ to find the corresponding $\overline{t}_{i}=\lpi{g}(y_{i})$.\footnote{Following the discussion in \Cref{sec:PseudoInverses}, $\lpib{S_g}$ is sufficient for this computation.}
The set $\left\{ t_{1},\dots,t_{n}\right\} \cup \left\{ \overline{t}_{1},\dots,\overline{t}_{m}\right\}$, preserving the ascending order, defines a finite set of breakpoints for $f \circ g$.
Then, we have again a finite set of points $\left(t_i, f(g(t_i))\right)$, and open intervals for which we compute $h$ as a segment with $\rho_h = \rho_f \cdot \rho_g$. 
In the example of \Cref{fig:comp-sequence-example}, $]0, 4[$ is such an interval:
\begin{itemize}
    \item for $S_g$ we find the set $\{t_1 = 1\}$
    \item for $S_f$ we find the set $\{y_1 = 1\} \rightarrow \left\{ \overline{t}_1 = \frac{1}{2} \right\}$
    \item the combined set of breakpoints is then $\left\{ \frac{1}{2}, 1 \right\}$, 
        and the open intervals that verify \Cref{eq:segment-comp} is $\left\{ \left]0, \frac{1}{2}\right[, \left]\frac{1}{2}, 1\right[, ]1, 4[ \right\}$
\end{itemize}

By generalizing this reasoning, we obtain \Cref{alg:pseudo-composition-seq}.

\begin{algorithm}[ht]
    \caption{Pseudo-code for the composition of finite sequences} \label{alg:pseudo-composition-seq}
    \hspace*{\algorithmicindent} \textbf{Input} Two finite sequences of elements, $ S_f $ of $f$ and $ S_g $ of $g$, so that $ S_g $ defined on $ \left[0, a\right[ $ and $ S_f $ defined on $ \left[g(0), g(a^-)\right]$.\\
    \hspace*{\algorithmicindent} \textbf{Return} Composition $ S_h = S_f \circ S_g $ consisting of a sequence of elements $ O = \left\{o_1, \dots, o_m\right\} $
    \begin{algorithmic}[1]
        \State Define an empty sequence of elements $ O \coloneqq \{\,\} $
        \State Let $ T $ be an empty, but ordered set of distinct rationals
        \State Let $ P_g $ be the set of points of $ S_g $
        \For{$ p_i $ \textbf{in} $ P_g $}
            \State Add the time $ t_i $ of $ p_i $ to $ T $
        \EndFor
        \State Let $ P_f $ be the set of points of $ S_f $, excluding the last point $g(a^-)$
        \For{$ \overline{p}_i $ \textbf{in} $ P_f $}
            \State Given time $ t_i $ of $ p_i $, add $ \overline{t}_i = \lpi{g}(t_i) $ to $ T $ \Comment{preserving the order in $ T $}
        \EndFor
        \For{each pair of consecutive $ (t_i, t_{i+1}) $ in $ T $}
            \State Append $ p \coloneqq \left(t_i, f(g(t_i))\right) $ to $ O $
            \State Append $ s \coloneqq \left(t_i, t_{i+1}, f\left(g(t_i^+)^+\right), f\left(g(t_{i+1}^-)^-\right)\right) $ to $ O $
        \EndFor
    \end{algorithmic}
\end{algorithm}

\subsection{By-curve algorithm for composition}
\label{subsec:Composition-Conclusions}

We can now discuss the by-curve algorithm, by combining the results presented in \Cref{subsec:Composition-Theorems,subsec:Composition-Algorithm}.
In \Cref{alg:pseudo-composition-curve} we show the pseudo-code to compute the composition $h = f \circ g$ of UPP functions $f$ and $g$, in the most general case.
The analogous for the more specialized cases, i.e., ultimately affine or ultimately constant operands, which here we omit for brevity, can be similarly derived by adjusting the parameters and domains computations.

\begin{algorithm}[ht]
    \caption{
        Pseudo-code for composition of UPP functions. 
    }
    \label{alg:pseudo-composition-curve}
    \hspace*{\algorithmicindent} \textbf{Input} Representation $R_f$ of a UPP function $f$, consisting of sequence $S_f$ and parameters $T_f$, $d_f$ and $c_f$; Representation $R_g$ of a non-negative and non-decreasing UPP function $g$, consisting of sequence $S_g$ and parameters $T_g$, $d_g$ and $c_g$ \\
    \hspace*{\algorithmicindent} \textbf{Return} Representation $R_{h}$ of $h = f \circ g$
    \begin{algorithmic}[1]
        \State Compute the UPP parameters for the result \Comment{\Cref{thm:UPP-composition}}
            \State\hspace{\algorithmicindent} $T_{h} \gets \max\left\{ \lpi{g}(T_{f}),T_{g}\right\}$
            \State\hspace{\algorithmicindent} $d_{h} \gets N_{d_{f}} \cdot d_{g}\cdot D_{c_{g}}$
            \State\hspace{\algorithmicindent} $c_{h} \gets D_{d_{f}}\cdot N_{c_{g}}\cdot c_{f}$
        \State Compute $S^{D_f}_f$ and $S^{D_g}_g$ \Comment{\Cref{eq:CompositionD}}
            \State\hspace{\algorithmicindent} $D_f \gets [g(0), g\left((T_h + d_h)^-\right)[$
            \State\hspace{\algorithmicindent} $S^{D_f}_f \gets \text{Cut}(R_f, D_f)$
            \State\hspace{\algorithmicindent} $D_g \gets [0, T_h + d_h[$
            \State\hspace{\algorithmicindent} $S^{D_g}_g \gets \text{Cut}(R_g, D_g)$
        \State Compute $S_{h} \gets S^{D_f}_f \circ S^{D_g}_g$ \Comment{\Cref{alg:pseudo-composition-seq}}
        \State $R_{h} \gets \left(S_{h}, T_{h}, d_{h}, c_{h}\right)$
    \end{algorithmic}
\end{algorithm}

Regarding the complexity of \Cref{alg:pseudo-composition-curve}, we note that the main cost is computing $S_{h} \gets S^{D_f}_f \circ S^{D_g}_g$.
Since \Cref{alg:pseudo-composition-seq} is a linear scan of $S^{D_f}_f$ and $S^{D_g}_g$, the resulting complexity is $\mathcal{O}\left(\card{S^{D_f}_f} + \card{S^{D_g}_g}\right)$.
Note that given the expressions in \Cref{thm:UPP-composition}, this computational cost highly depends on the numerical properties of the operands, i.e., numerators and denominators of UPP parameters, rather than simply the cardinalities of $R_f$ and $R_g$.
Thus, using the specialized properties of \Cref{prop:UPP-composition-g-UA,prop:UPP-composition-f-UA,prop:UA-composition} yields performance improvements, since $D_f$ and $D_g$ are smaller.

We remark again that the result of the composition may have a non-minimal representation (see the discussion at the end of \Cref{subsec:PseudoInverses-Theorems}). 

\section{Proof of Concept}
\label{sec:PoC}

In this section, we show how the algorithms computed in this paper allow one to replicate the result appeared in a recent NC paper \cite{Tabatabaee2021IWRR}.

The algorithms described in this paper, including variants and corner cases omitted for brevity, are implemented in the publicly available Nancy NC library \cite{Nancy22}.
Nancy is a C\# library implementing the UPP model and its operators, as described in \cite{bouillard2008algorithmic,Bouillard2018DNC}. 
Moreover, it implements state-of-the-art algorithms that improve the efficiency of NC operators, described in \cite{ZippoEtAlii2022}, and lower pseudo-inverse, upper pseudo-inverse and composition operators, described in this paper. 
Nancy makes extensive use of parallelism. 
However, the NC operators described in this paper are implemented as sequential.

As a notable example, we implemented the results from \cite[Theorem~1]{Tabatabaee2021IWRR}, which uses the composition operator, using the same parameters of the example in \cite[Figure 3]{Tabatabaee2021IWRR}. 
The above theorem allows us to compute the service curve for a flow served by an Interleaved Weighted Round-Robin scheduler, once a) the weight of the flow; b) the minimum and maximum packet length for each flow, and c) the (strict) service curve for the entire aggregate of flows $\beta(t)$ are known. 
The complete formulation of the theorem -- which is rather cumbersome -- is postponed to \Cref{app-sec:SC-IWRR}. 
For the purpose of this proof of concept, the important bit is that computing the service curve of the flow involves computing a function $\gamma_i$ that takes into account flow $i$'s characteristics (e.g., weight, packet lengths), and then, given $\beta$ as the (strict) service curve of the server regulated by IWRR, compute the (strict) per-flow service curve for flow $i$ as $\beta_i = \gamma_i \circ \beta$.\footnote{Recall that composition requires the lower pseudo-inverse of the inner function to be computed, hence this example makes use of both the algorithms presented in this paper.}
In the example in \cite[Figure 3]{Tabatabaee2021IWRR}, $\beta$ is a constant-rate service curve, thus UA, while $\gamma_i$ is, in general, a UPP function. 
On the one hand, this confirms that limiting NC algorithms to UA curves only is severely constraining -- in this example, one could not compute flow $i$'s service curve without an algorithm that handles UPP curves. 
On the other hand, it means that we can obtain the same result by applying both \Cref{thm:UPP-composition} and its specialized version for UA inner functions \Cref{prop:UPP-composition-g-UA}, and that we can expect the latter to be more efficient due to the tighter $D_f$, as explained below \Cref{eq:CompositionD-g-UA}.

Our experiments confirm the above intuition. We run the computation on a laptop computer (i7-10750H, 32 GB RAM).
As shown in \Cref{table:iwrr_benchmark}, when using \Cref{thm:UPP-composition}, computing the result took a median of 1.11 seconds. 
On the other hand, using \Cref{prop:UPP-composition-g-UA} the same result is obtained in 0.55 milliseconds in the median, an improvement of three orders of magnitude.
\Cref{fig:iwrr_code,fig:iwrr_beta_i} report the code used and the resulting plot. 

\begin{table}[t]
	\centering
	\caption{Performance comparison of composition with and without UA optimization.}
	\label{table:iwrr_benchmark}
	\begin{tabular}{|l|c|c|}
		\hline
		Runtime & Not optimized & Optimized \\ \hline
		75th percentile & 1117.72 ms & 0.67 ms \\ \hline
		median & 1105.01 ms & 0.55 ms \\ \hline
		25th percentile & 1088.61 ms & 0.50 ms \\ \hline
	\end{tabular}
\end{table}

It is worth noting that \cite[Theorem~1]{Tabatabaee2022DRR} describes a similar result for the Deficit Round-Robin scheduler, under similar hypotheses, still making use of composition, with the outer curve being non-UA. 
The derivations in this section apply to this case as well, with minimal obvious modifications. 
Several other results in \cite{Tabatabaee2022DRR} make use of composition as well.

\begin{figure}
    \centering
    \includegraphics[width=0.87\linewidth]{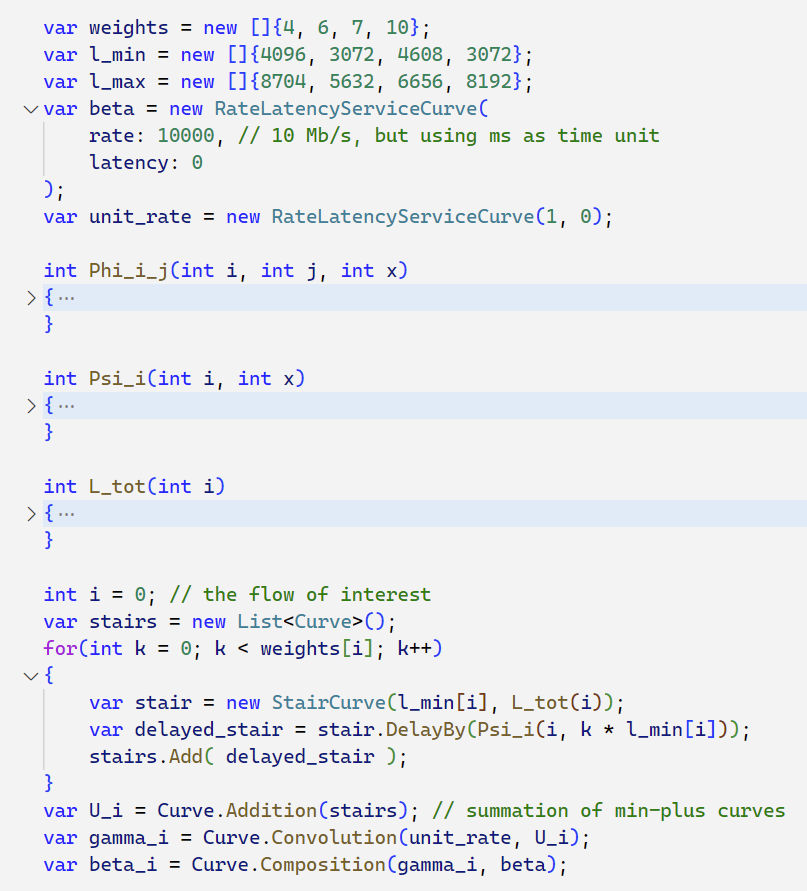}
    \caption{Code used to replicate the results of \cite[Theorem~1]{Tabatabaee2022DRR}.}
    \label{fig:iwrr_code}
\end{figure}

\begin{figure}
    \centering
    \includegraphics[width=0.92\linewidth]{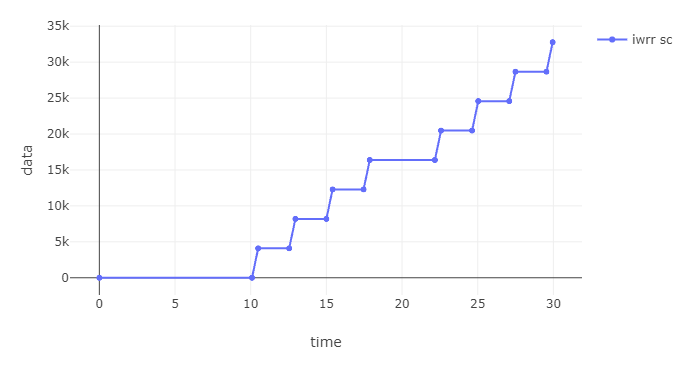}
    \caption{Plot of the resulting service curve $\beta_i$.}
    \label{fig:iwrr_beta_i}
\end{figure}

\section{Conclusions}
\label{sec:Conclusions}

Automated computation of Network Calculus operations is necessary to carry out analyses of non-trivial network scenarios. 
Therefore, algorithms that transform representations of operand functions into result functions are required for each "useful" NC operator. 
Recently, pseudo-inverses and composition operators have been repeatedly use in NC papers. 
To the best of our knowledge, these operators lacked an algorithmic description that would allow their implementation in software. 

This paper fills the above gap, by providing algorithms for lower and upper pseudo-inverses and composition of operators. 
We have presented algorithms that work under general assumptions (i.e., UPP operands), as well as specialized ones that leverage the fact that operands are UA (or UC) to compute results faster. 
We have discussed the complexity of these algorithms, as well as corner cases, with a rigorous mathematical exposition. 
Beside the theoretical contribution, we provided a practical one by including the above algorithms (together with the others known from the literature) in an open-source free library called Nancy. 
This allows researchers  to experiment with our results to study complex scenarios, or support the generation of novel theoretical insight.

Future work on this topic will include studying the computational and numerical properties of NC operators. 
As the example in \Cref{sec:PoC} shows (not to mention those in \cite{ZippoEtAlii2022}, by the same authors), exploiting more knowledge on the operands allows one to compute the same results via specialized versions of the algorithms, often in considerably shorter times (some by orders of magnitude). 
We believe that this is an avenue of research worth pursuing, with the aim of enabling larger-scale real-world performance studies.

\section{Acknowledgements}

This work was partially supported by the Italian Ministry of Education and Research (MIUR) in the framework of the CrossLab project (Departments of Excellence), and by the University of Pisa, through grant “Analisi di reti complesse: dalla teoria alle applicazioni” - PRA 2020.

\FloatBarrier

\bibliography{biblio.bib}

\newpage
\appendix

\section{Properties of Ultimately Affine Functions}
\label{app-sec:properties-UAF}

    
    
\begin{prop}
    \label{prop:UA-equivalence} 
    A function $f$ is ultimately affine iff there exist $T\in\mathbb{Q}^{+},\sigma,\rho\in\mathbb{Q}$ such that 
    \begin{equation}
        f(t)\in\left\{ \rho t+\sigma,-\infty,+\infty\right\} \quad\forall t>T.\label{eq:UA-property-toolbox}
    \end{equation}
\end{prop}
    
\begin{proof}
    ``$\Rightarrow$''\\
    Let $f$ be ultimately affine. Define $T=T_{f}^{a},$ $\sigma=f(T_{f}^{a})$ and $\rho=\rho_{f}$. 
    Then, it holds for all $t>T$ that
    \begin{equation*}
        f(t)=f\left(t-T_{f}^{a}+T_{f}^{a}\right)=f(T_{f}^{a})+\rho_{f}\cdot t=\sigma+\rho t.
    \end{equation*}
    
    ``$\Leftarrow$''
    
    Assume $f$ to verify the condition in Equation~(\ref{eq:UA-property-toolbox}).
    If $f(t)\in\left\{ -\infty,+\infty\right\} ,$ the proof is trivial.
    Let us therefore assume that $f(t)=\rho t+\sigma$ for all $t>T$.
    Define $T_{f}^{a}=T,\rho_{f}=\rho$. Then for all $t>0$
    \begin{equation*}
    f(t+T_{f}^{a})=\sigma+\rho_{f}\left(t+T_{f}^{a}\right)=\left(\rho_{f}T_{f}^{a}+\sigma\right)+\rho_{f}t=f(T_{f}^{a})+\rho_{f}t.
    \end{equation*}
    This concludes the proof.
\end{proof}

\section{Calculation of Lower Pseudo-Inverses}
\label{app-sec:PseudoInverses-Formalization}

We report here the rigorous mathematical derivations for cases c1-c8 in \Cref{table:cases}. 

\subsection*{Point after segment (cases c1-c4)}

In these cases we have, in general, an $f$ such that
\begin{equation*}
    f(x)=\begin{cases}
        b_{1} + \rho \left(x - t_{1}\right), & \text{if }t_{1}<x<t_{2},\\
        b_{2}, & \text{if }x\to t_{2}^{-}\\
        b_{3} & \text{if }x=t_{2}
    \end{cases}
\end{equation*}
Since $f$ is non-decreasing, $b_{1} + \rho \left(x - t_{1}\right)\leq b_{2}\leq b_{3}$ for all $x\in\left]t_{1},t_{2}\right[$. 

We then classify based on two properties:
\begin{itemize}
    \item Whether the segment is constant, i.e., $\rho = 0 \rightarrow b_1 = b_2$
    \item Whether there is a discontinuity at $t_2$, i.e., $b_2 < b_3$
\end{itemize}

\paragraph{Case c1: $\rho=0$ and $b_{1}=b_{2}<b_{3}$ (constant segment followed by a discontinuity). }

It holds that
\begin{equation}
    \lpi{f}(y)=\begin{cases}
        \inf\left\{ x\mid f(x)\geq y\right\} =t_{2}, & \text{if }b_{1}<y<f(t_{2})=b_{2},\\
        \inf\left\{ x\mid f(x)\geq\underbrace{y}_{=b_{2}}\right\} =t_{2}, & \text{if }y=f(t_{2})=b_{2}.
    \end{cases}
    \label{eq:lpi-sp-const-jump}
\end{equation}
and
\begin{equation}
    \upi{f}(y)=\begin{cases}
        \sup\left\{ x\mid f(x)\leq\overbrace{y}^{=b}\right\} =t_{2}, & \text{if }y=b_{1}=f(t_{1}^{+}),\\
        \sup\left\{ x\mid f(x)\leq y\right\} =t_{2}, & \text{if }b_{1}=f(t_{1}^{+})<y<f(t_{2}),\\
        \inf\left\{ x\mid\underbrace{f(x)}_{=b_{2}}>y\right\} =\sup\left\{ x\mid f(x)\leq y\right\} =t_{2} & \text{if }y=f(t_{2})=b_{2}.
    \end{cases}
    \label{eq:upi-sp-const-jump}
\end{equation}

\paragraph{Case c2: $\rho=0$ and $b_{1}=b_{2}=b_{3}$ (constant segment without any discontinuity). }

It holds that
\begin{equation}
    f^{\downarrow}(y)=\inf\left\{ x\mid\overbrace{f(x)}^{=b_{1}}\geq y\right\} =t_{1}, \text{if }y=b_{1}
    \label{eq:lpi-sp-const-no-jump}
\end{equation}
(however, we do not add a value as it is processed in the ``segment after point''-section) and
\begin{equation}
    \upi{f}(y)\coloneqq
    \sup\left\{ x\mid\underbrace{f(x)}_{=b_{1}}\leq y\right\} =t_{1}, \text{if }y=b_{1}
    \label{eq:upi-sp-const-no-jump}
\end{equation}

\paragraph{Case c3: $\rho>0$ and $b_{2}<b_{3}$ (non-constant segment followed by a discontinuity).}

\begin{equation}
    \lpi{f}(y)=\begin{cases}
        \inf\left\{ x\mid f(x)\geq y\right\} =t_{2}, & \text{if }y=b_{1}+r\left(t_{2}-t_{1}\right)=f(t_{2}^{-})\\
        \inf\left\{ x\mid f(x)\geq y\right\} =t_{2}, & \text{if }f(t_{2}^{-})<y<f(t_{2})=b_{3}\\
        \inf\left\{ x\mid f(x)\geq\underbrace{y}_{b_{3}}\right\} =t_{2}, & y=f(t_{2})=b_{3}
    \end{cases}
    \label{eq:lpi-sp-inc-jump}
\end{equation}
and
\begin{equation}
    \upi{f}(y)=\begin{cases}
        \inf\left\{ x\mid f(x)>y\right\} =t_{2}, & \text{if }y=b_{1}+r\left(t_{2}-t_{1}\right)=f(t_{2}^{-})\\
        \inf\left\{ x\mid f(x)>y\right\} =t_{2}, & \text{if }f(t_{2}^{-})<y<f(t_{2})=b_{3}\\
        \inf\left\{ x\mid f(x)>\underbrace{y}_{b_{3}}\right\} =t_{2}, & y=f(t_{2})=b_{3}.
    \end{cases}
    \label{eq:upi-sp-inc-jump}
\end{equation}

\paragraph{Case c4: $\rho>0$ and $b_{2}=b_{3}$ (non-constant segment without any discontinuity).}

\begin{equation}
    \lpi{f}(y)=\inf\left\{ x\mid f(x)\geq\underbrace{y}_{b_{2}}\right\} =t_{2},  y=f(t_{2})=b_{2}
    \label{eq:lpi-sp-inc-no-jump}
\end{equation}
and
\begin{equation}
    \upi{f}(y)=\inf\left\{ x\mid f(x)>\underbrace{y}_{b_{2}}\right\} =t_{2},  y=f(t_{2})=b_{2}.
    \label{eq:upi-sp-inc-no-jump}
\end{equation}

\subsection*{Segment after point (cases 5-8)}

In these cases we have, in general, an $f$ such that
\begin{equation*}
f(x)=\begin{cases}
    b_{1}, & x=t_{1},\\
    b_{2}, & x\to t_{1}^{+},\\
    b_{2} + \rho \left(x-t_{1}\right), & t_{1}<x<t_{2}.
\end{cases}
\end{equation*}

We then classify based on two properties:
\begin{itemize}
    \item Whether the segment is constant, i.e., $\rho = 0 \rightarrow b_2 = b_3$
    \item Whether there is a discontinuity at $t_1$, i.e., $b_1 \neq b_2$
\end{itemize}

\paragraph{Case c5: $b_1<b_2=b_3$ and $\rho=0$ (discontinuity followed by a constant segment). }

It holds that
\begin{equation}
    \lpi{f}(y)=\begin{cases}
        \inf\left\{ x\mid f(x)\geq y\right\} =\sup\left\{ x\mid f(x)<y\right\} = t_{1}, & \text{if } b_1 < y < f(t_{1}^{+}) = b_2,\\
        \inf\left\{ x\mid f(x)\geq\underbrace{y}_{=b_1}\right\} = t_{1}, & \text{if } y = f(t_{1}^{+}) = b_2
    \end{cases}
    \label{eq:lpi-ps-const-jump}
\end{equation}
and
\begin{equation}
    \upi{f}(y)=\begin{cases}
        \inf\left\{ x\mid f(x)>y\right\} =\sup\left\{ x\mid f(x)\leq y\right\} = t_{1}, & \text{if } b_1 < y < f(t_{1}^{+}) = b_2,\\
        \inf\left\{ x\mid f(x)>y\right\} =\sup\left\{ x\mid f(x)\leq\underbrace{y}_{=b_1}\right\} = t_2, & \text{if } y = f(t_{1}^{+}) = b_2
    \end{cases}
    \label{eq:upi-ps-const-jump}
\end{equation}

\paragraph{Case c6: $b_1=b_2=b_3$ and $\rho=0$ (no discontinuity and a constant segment). }

Then it holds that
\begin{equation}
    \lpi{f}(y)=\inf\left\{ x\mid\overbrace{f(x)}^{=b_{1}}\geq y\right\} =t_{1}, \text{if }y=b_{1}
    \label{eq:lpi-ps-const-no-jump}
\end{equation}
(however, we do not add a value as it is processed in the ``point after segment''-section) and
\begin{equation}
    \upi{f}(y)\coloneqq \sup\left\{ x\mid\underbrace{f(x)}_{=b_{1}}\leq y\right\} =t_{2},  \text{if }y=b_{1}.
    \label{eq:upi-ps-const-no-jump}
\end{equation}

\paragraph{Case c7: $b_1<b_2$ and $\rho>0$ (discontinuity followed by a non-constant segment).}

Then, 
\begin{equation}
    \lpi{f}(y)=\begin{cases}
        \inf\left\{ x\mid f(x)\geq y\right\} = t_{1}, & \text{if }b_{1}<y<f(t_{1}^{-})=b_{2},\\
        \inf\left\{ x\mid f(x)\geq\underbrace{y}_{=b_2}\right\} = t_{1}, & \text{if } y = f(t_{1}^{-}) = b_{2}\\
        \inf\left\{ x\mid b_{2} + \rho \left(x - t_1\right) \geq y\right\} = \frac{y-b_{2}}{\rho}, & \text{if } b_{2} < y < b_{3}
    \end{cases}
    \label{eq:lpi-ps-inc-jump}
\end{equation}
and
\begin{equation}
    \upi{f}(y)=\begin{cases}
        \inf\left\{ x\mid f(x)>y\right\} = t_{1}, & \text{if } b_1 < y < b_2,\\
        \inf\left\{ x\mid f(x)>y\right\} =\sup\left\{ x\mid f(x)\leq\underbrace{y}_{=b_2}\right\} =t_{1}, & \text{if }y = b_2\\
        \inf\left\{ x\mid b_{2} + \rho \left(x - t_1\right) \right\} =\frac{y-b_2}{\rho}, & \text{if } b_2 < y < b_3.
    \end{cases}
    \label{eq:upi-ps-inc-jump}
\end{equation}

\paragraph{Case c8: $b_1=b_2$ and $\rho>0$ (no discontinuity and non-constant segment). }

Then, 
\begin{equation}
    \lpi{f}(y)=\inf\left\{ x\mid b_{1} + \rho \left(x - t_1\right) \geq y\right\} =\frac{y-b_{1}}{\rho}, \text{if } b_{1} = f(t_{1}^{+})<y<f(t_{2}^{-}) = b_{2}
    \label{eq:lpi-ps-inc-no-jump}
\end{equation}
and
\begin{equation}
    \upi{f}(y)=\inf\left\{ x\mid b_{1} + \rho \left(x - t_1\right) > y\right\} =\frac{y-b_{1}}{\rho}, \text{if } b_1 = f(t_{1}^{+})<y<f(t_{2}^{-}) = b_{2}.
    \label{eq:upi-ps-inc-no-jump}
\end{equation}

\section{Composition of Ultimately Constant Functions}
\label{app-sec:comp-UCF}

\begin{prop}
    Let $f$ and $g$ be two ultimately pseudo-periodic (UPP) functions that are not wUI with $g$ being non-negative, non-decreasing and ultimately constant (UC).
    Then, their composition $h\coloneqq f\circ g$ is again UC with
    \begin{equation}
        T_{h} = T_{g}.
    \end{equation}
\end{prop}
\begin{proof}
    For $ t \geq T_g $, it holds that
    \begin{equation*}
        h(t) = f(g(t)) = f(g(T_g)) = h(T_h).
    \end{equation*}
\end{proof}

\begin{prop}
    Let $f$ be UC and $g$ be UPP function that is non-negative and non-decreasing. 
    Then, their composition $h\coloneqq f\circ g$ is again UC with
    \begin{equation}
        T_{h} = \lpi{g}(T_f).
    \end{equation}
\end{prop}
\begin{proof}
    For $ t \geq \lpi{g}(T_f) $, it holds that
    \begin{equation*}
        h(t) = f(g(t)) = f(T_f) = h(T_h).
    \end{equation*}
\end{proof}

\begin{prop}
    Let $f$ and $g$ be UC functions with $g$ being non-negative, non-decreasing. 
    Then, their composition $h\coloneqq f\circ g$ is again UC with
    \begin{equation}
        T_{h} = \min\left\{T_f, \lpi{g}(T_f)\right\}.
    \end{equation}
\end{prop}
\begin{proof}
    The proof is simply a combination of the previous two propositions.
\end{proof}

\section{Service Curve of a Flow in Interleaved Weighted Round Robin}
\label{app-sec:SC-IWRR}

We report here the statement of Theorem~1 in \cite{Tabatabaee2021IWRR}, for ease of reference. 
We slightly rephrased it to aid comprehension. 

\begin{thm}[Strict Per-Flow Service Curves for IWRR]
	Assume $n$ flows arriving at a server performing interleaved weighted round robin (IWRR) with weights $w_{1},\dots,w_{n}$.
    Let $ l_i^{\min} $ and $ l_j^{\max} $ denote the minimum and maximum packet size of the respective flow.
	Let this server offer a super-additive strict service curve $\beta$ to these $n$ flows. 
	Then, 
	\begin{equation*}
		\beta^i(t) \coloneqq \gamma_i\left(\beta(t)\right),
	\end{equation*}
	is a strict service curve for flow $f_{i}$, where
	\begin{align*}
		\gamma_i(t)  &\coloneqq \conv{\beta_{1,0}}{U_i}{t}, \nonumber\\
		U_i (t)  &\coloneqq  \sum_{k=0}^{w_i - 1} \nu_{l_i^{\min},L_\mathrm{tot}}\left(\posPart{t - \psi\left(k l_i^{\min}\right)}\right), \nonumber\\
		L_\mathrm{tot}  &\coloneqq  w_i l^{\min}_i + \sum_{j: j \neq i} {w_j l^{\max}_j}, \nonumber\\
		\psi_i(x)  &\coloneqq  x + \sum_{j\neq i} \phi_{ij} \left(\left\lfloor \frac{x}{l_{i}^{\min}}\right\rfloor \right) l_j^{\max}, \nonumber\\
		\phi_{ij}(p)  &\coloneqq  \left\lfloor \frac{p}{w_{i}}\right\rfloor w_{j}+\posPart{w_{j}-w_{i}} +\min\left\{ \left(p\mod w_{i}\right)+1,w_{j}\right\}, \label{eq:phi-ij}
	\end{align*}
    $\beta_{1,0}$ is a constant-rate function with slope 1, and the stair function $ \nu_{h,P}(t) $ is defined as 
    \begin{equation*}
        \nu_{h, P}(t)\coloneqq h\left\lceil \frac{t}{P}\right\rceil, \qquad\text{for }t\geq0.
        \label{eq:stair-function}
    \end{equation*}
\end{thm}

\end{document}